\documentclass[article,12pt]{amsart}

\usepackage{amsfonts}
\usepackage{amsmath}
\usepackage{amssymb}
\usepackage{graphicx}
\usepackage{fullpage}
\usepackage{enumitem}
\usepackage{bbm}
\usepackage{stackengine}
\usepackage{stmaryrd}
\usepackage{color}
\usepackage{multirow}

\allowdisplaybreaks

\numberwithin{equation}{section}

\theoremstyle{plain}

\newtheorem{cor}{Corollary}[section]
\newtheorem{prop}{Proposition}[section]
\newtheorem{defi}{Definition}[section]

\theoremstyle{remark}
\newtheorem{rem}{Remark}[section]
\newtheorem{ex}{Example}[section]

\newcommand{\cB}{\mathcal{B}}
\newcommand{\cE}{\mathcal{E}}
\newcommand{\cG}{\mathcal{G}}
\newcommand{\cI}{\mathcal{I}}
\newcommand{\cK}{\mathcal{K}}
\newcommand{\cM}{\mathcal{M}}
\newcommand{\cN}{\mathcal{N}}
\newcommand{\cT}{\mathcal{T}}

\newcommand{\cW}{\mathcal{W}}
\newcommand{\cGIG}{\cG\cI\cG}
\newcommand{\cMGIG}{\cM\cG\cI\cG}
\newcommand{\cMN}{\cM\cN}

\newcommand{\dC}{\mathbb{C}}
\newcommand{\dE}{\mathbb{E}}
\newcommand{\dP}{\mathbb{P}}
\newcommand{\dR}{\mathbb{R}}
\newcommand{\dS}{\mathbb{S}}
\newcommand{\dV}{\mathbb{V}}
\newcommand{\dX}{\mathbb{X}}
\newcommand{\dY}{\mathbb{Y}}
\newcommand{\gdX}{\barbelow{\dX}}
\newcommand{\diag}{\textnormal{diag}}

\newcommand{\dd}{\mathrm{d}}
\newcommand{\de}{\mathrm{e}}

\newcommand{\Oy}{\Omega_{y}}
\newcommand{\oy}{\omega_{y}}
\newcommand{\Oyi}{\Oy^{-1}}
\newcommand{\Ox}{\Omega_{x}}
\newcommand{\D}{\Delta}
\newcommand{\Dt}{\wt{\D}}
\newcommand{\gD}{\barbelow{\D}}

\newcommand{\indep}{\protect\mathpalette{\protect\independenT}{\perp}}
\newcommand\barbelow[1]{\stackunder[1.2pt]{$#1$}{\rule{.8ex}{.075ex}}}
\def\independenT#1#2{\mathrel{\rlap{$#1#2$}\mkern2mu{#1#2}}}
\newcommand{\simindep}{\overset{\indep}{\sim}}
\newcommand{\hsp}{\hspace{0.5cm}}
\newcommand{\ind}{\mathbbm{1}}
\newcommand{\tr}{\textnormal{tr}}
\renewcommand{\vec}{\textnormal{vec}}
\newcommand{\card}{\textnormal{Card}}

\newcommand{\wt}{\widetilde}
\newcommand{\corr}{\dC\textnormal{orr}}
\newcommand{\limp}{~ \overset{\dP}{\longrightarrow} ~}

\def\leq{\leqslant}
\def\geq{\geqslant}

\begin{document}

\title[A Bayesian PGGM with sparsity]
{A Bayesian approach for partial Gaussian graphical models with sparsity
\vspace{2ex}}
\author[E. Okome Obiang]{Eunice Okome Obiang}
\address{Univ Angers, CNRS, LAREMA, SFR MATHSTIC, F-49000 Angers, France.}
\email{okome@math.univ-angers.fr}
\author[P. J\'ez\'equel]{Pascal J\'ez\'equel}
\address{1 Unit\'e de Bioinfomique, Institut de Canc\'erologie de l'Ouest, Bd Jacques Monod, 44805 Saint Herblain Cedex, France.\vspace{-2ex}}
\address{2 SIRIC ILIAD, Nantes, Angers, France. \vspace{-2ex}}
\address{3 CRCINA, INSERM, CNRS, Universit\'e de Nantes, Universit\'e d'Angers, Institut de Recherche en Sant\'e-Universit\'e de Nantes, 8 Quai Moncousu - BP 70721, 44007, Nantes Cedex 1, France.}
\email{pascal.jezequel@ico.unicancer.fr}
\author[F. Pro\"ia]{Fr\'ed\'eric Pro\"ia}
\address{Univ Angers, CNRS, LAREMA, SFR MATHSTIC, F-49000 Angers, France.}
\email{frederic.proia@univ-angers.fr}

\thanks{}

\begin{abstract}
We explore various Bayesian approaches to estimate partial Gaussian graphical models. Our hierarchical structures enable to deal with single-output as well as multiple-output linear regressions, in small or high dimension, enforcing either no sparsity, sparsity, group sparsity or even sparse-group sparsity for a bi-level selection through partial correlations (direct links) between predictors and responses, thanks to spike-and-slab priors corresponding to each setting. Adaptative and global shrinkages are also incorporated in the Bayesian modeling of the direct links. An existing result for model selection consistency is reformulated to stick to our sparse and group-sparse settings, providing a theoretical guarantee under some technical assumptions. Gibbs samplers are developed and a simulation study shows the efficiency of our models which give very competitive results, especially in terms of support recovery. To conclude, a real dataset is investigated.
\end{abstract}

\maketitle

\vspace{-0.5cm}

\begin{center}
\end{center}

\medskip

\section{Introduction and Motivations}
\label{SecIntro}

This paper is devoted to the Bayesian estimation of the partial Gaussian graphical models. Graphical models are now widespread in many contexts, like image analysis, economics or biological regulation networks, neural models, etc. A graphical model for the $d$-dimensional Gaussian vector $Z \sim \cN_d(\mu, \Sigma)$ is a model where the conditional dependencies between the coordinates of $Z$ are represented by means of a graph. We refer the reader to the handbook recently edited by Maathuis \textit{et al.} \cite{MaathuisEtAl18} for a very complete survey of graphical models theory, or to Chap. 7 of Giraud \cite{Giraud14} for a wide introduction to the subject. It is well-known that the partial correlation between $Z_i$ and $Z_j$ satisfies
\begin{equation*}
\corr(Z_i,\, Z_j\, \vert\, Z_{\neq\, i,\, j}) = -\frac{\Omega_{ij}}{\sqrt{\Omega_{ii}\, \Omega_{jj}}}
\end{equation*}
where $\Omega = \Sigma^{-1} \in \dS_{++}^{\, d}$ is the precision matrix of $Z$ (the notation $\dS_{++}^{\, d}$ for the cone of symmetric positive definite matrices of dimension $d$ is used in all the paper). A fundamental consequence of this is that there is a partial correlation between $Z_i$ and $Z_j$ if and only if the $(i,j)$-th element of $\Omega$ is non-zero. The sparse estimation of $\Omega$ is therefore a major issue for variable selection in high-dimensional studies, which has given rise to a substantial literature, see \textit{e.g.} the seminal work of Meinshausen and B\"uhlmann \cite{MeinshausenBuhlmann06}. This logically led numerous authors to investigate interesting properties under various kind of hypotheses, estimation procedures and penalties. Let us mention for example the optimality results obtained by Cai and Zhou \cite{CaiZhou12} and the penalized estimations of Yuan and Lin \cite{YuanLin07}, Rothman \textit{et al.} \cite{RothmanEtAl08}, Banerjee \textit{et al.} \cite{BanerjeeEtAl08}, Cai \textit{et al.} \cite{CaiEtAl11} or Ravikumar \textit{et al.} \cite{RavikumarEtAl11}, all coming with theoretical guarantees, algorithmic considerations and real world examples. Besides, the famous graphical Lasso of Friedman \textit{et al.} \cite{FriedmanEtAl08} has become an essential tool for dealing with precision matrix estimation. Perhaps more attractive to us since focusing on each entry of the precision matrix (no longer taken as a whole), the approach of Ren \textit{et al.} \cite{RenEtAl15} is remarkable and will serve as a basis for comparison in our simulation study. The Bayesian inference counterpart has been developed as well, it is \textit{e.g.} the subject of Chap. 10 of Maathuis \textit{et al.} \cite{MaathuisEtAl18} where various Wishart-type priors are considered for $\Omega$, see also Li \textit{et al.} \cite{LiEtAl19} or Gan \textit{et al.} \cite{GanEtAl19} for spike-and-slab approaches and all references within.

\smallskip

Suppose now that we deal with a multivariate linear regression of the form
\begin{equation*}
\dY = \dX B + E
\end{equation*}
where $\dY \in \dR^{n \times q}$ is a matrix of $q$-dimensional responses of which $k$-th row is $Y_k^{\, t}$, $\dX \in \dR^{n \times p}$ is a matrix of $p$-dimensional predictors of which $k$-th row is $X_k^{\, t}$, $B \in \dR^{p \times q}$ contains the regression coefficients and $E \in \dR^{n \times q}$ is a matrix-variate Gaussian noise. The Partial Gaussian Graphical Model (PGGM), developped \textit{e.g.} by Sohn and Kim \cite{SohnKim12} or Yuan and Zhang \cite{YuanZhang14}, appears as a powerful tool to exhibit relationships between predictors and responses that exist through partial correlations (called from now on `direct links', as opposed to `indirect links' resulting from correlations). Indeed, assume that the couple $(Y_k, X_k) \in \dR^{q+p}$ is jointly normally distributed with zero mean, covariance $\Sigma$ and precision $\Omega$. Then, the block decomposition given by
\begin{equation*}
\Omega = \begin{pmatrix}
\Oy & \D \\
\D^t & \Ox
\end{pmatrix}
\end{equation*}
with $\Oy \in \dS_{++}^{\, q}$, $\D \in \dR^{q \times p}$ and $\Ox \in \dS_{++}^{\, p}$ leads to $Y_k\, \vert\, X_k\, \sim\, \cN_q(-\Oyi \D\, X_k,\, \Oyi)$. This is a crucial remark because one can see that the multiple-output regression $Y_k = B^{\, t}\, X_k + E_k$ with Gaussian noise $E_k \sim \cN_q(0, R)$ may be reparametrized with
\begin{equation}
\label{Reparam}
B = -\D^t\, \Oyi \hsp \text{and} \hsp R = \Oyi.
\end{equation}
A large volume of literature exists for the sparse estimation of $B$ with arbitrary group structures (see \textit{e.g.} Li \textit{et al.} \cite{LiEtAl15} or Chap. 6 of Giraud \cite{Giraud14}), but we will not tackle this issue in our study. At least not frontally but indirectly, since the latter relations show that an estimation of $B$ is possible through the one of the pair $(\Oy, \D)$. Whereas $B$ contains direct and indirect links between the predictors and the responses (due \textit{e.g.} to strong correlations among the variables), $\D$ is clearly more interesting from an inferential point of view for it only contains direct links. However, while the estimation of $(\Oy, \D)$ appears to be essential, it usually depends on the accuracy of the estimation of the whole precision matrix, which, in turn, may be strongly affected by the one of $\Ox$. For example, the graphical Lasso of Friedman \textit{et al.} \cite{FriedmanEtAl08} involves maximizing the log-likelihood penalized by the elementwise $\ell_1$ norm of $\Omega$. For multiple-output high-dimensional regressions where generally $p \gg q$, we understand that a significant bias is likely to result from the large-scale shrinkage. Another substantial advantage of the partial model is that we can override this issue by computing a new objective function in which $\Ox$ has disappeared, \textit{i.e.} the penalized log-likelihood
\begin{eqnarray}
\label{LikPGGM}
L_{n}(\Oy, \D) ~ = ~ -\ln \det(\Oy) + \tr(S_y\, \Oy) + 2\, \tr(S_{yx}^{\, t}\, \D) \hsp \hsp \hsp \hsp \hsp \hsp \nonumber \\
\hsp \hsp \hsp \hsp \hsp \hsp +~ \tr(S_x\, \D^t\, \Oyi\, \D) + \lambda\, \textnormal{pen}(\Oy) + \mu\, \textnormal{pen}(\D)
\end{eqnarray}
where $S_x \in \dS_{++}^{\, p}$ and $S_y \in \dS_{++}^{\, q}$ are the empirical variances of the responses and the predictors, respectively, and where $S_{yx} \in \dR^{q \times p}$ is the empirical covariance, computed on the basis of a set of $n$ observations. This can be obtained either by considering the multiple-output Gaussian regression scheme, or, as it is done by Yuan and Zhang \cite{YuanZhang14}, by eliminating $\Ox$ thanks to a first optimization step in the objective function of the graphical model. The usual convex penalties are elementwise $\ell_1$ norms, possibly deprived of the diagonal terms for $\Oy$. This paved the way to the recent study of Chiquet \textit{et al.} \cite{ChiquetEtAl17} where the authors replace the penalty on $\Oy$ by a structuring one enforcing various kind of sparsity patterns in $\D$, and to the one of Okome Obiang \textit{et al.} \cite{OkomeEtAl21} in which some theoretical guarantees are provided for a slightly more general estimation procedure.

\smallskip

However, to the best of our knowledge, the Bayesian approach for the PGGM is a new research topic. Given the outputs gathered in $\dY$ and the predictors gathered in $\dX$, the objective of this paper is the Bayesian estimation of the direct links and the precision matrix of the responses. This is inspired by the ideas of Xu and Ghosh \cite{XuGosh15} for the single-output setting ($q=1$), and by the ones of Liquet \textit{et al.} \cite{LiquetEtAl17} for the multiple-output setting ($q > 1$). Taking advantage of the relations \eqref{Reparam}, we consider that a Gaussian prior for $B$ must remain Gaussian for $\D$ (with a correctly updated variance), and that an inverse Wishart prior for $R$ merely becomes a Wishart one for $\Oy$. Yet, despite these seemingly small changes in the design of the priors, we will see that the resulting distributions are completely different. The hierarchical models that we are going to study all come from this working base, but let us point out that a wide variety of refinements exist in the recent literature for Bayesian sparsity, like the grouped `horseshoe' of Xu \textit{et al.} \cite{XuEtAl16}, the `aggressive' multivariate Dirichlet-Laplace prior of Wei \textit{et al.} \cite{WeiEtAl20}, the theoretical results for group selection consistency of Yang and Narisetty \cite{YangNarisetty20} or even the extension of the Bayesian spike-and-slab group selection to generalized additive models of Bai \textit{et al.} \cite{BaiEtAl20}, all related to the regression setting but that might also be investigated for PGGMs. To enforce various types of sparsity in $\D$ for high-dimensional problems, we decided to make use of spike-and-slab priors, with a spike probability guided by a conjugate Beta distribution.

\smallskip

The paper is organized as follows. Sections \ref{SecS}, \ref{SecGS} and \ref{SecSGS} are dedicated to the study of our hierarchical models enforcing either no sparsity, sparsity, group sparsity or sparse-group sparsity in the direct links, respectively, according to the terminology of Sec. 2.1 of Giraud \cite{Giraud14}. In particular, we will see that our bi-level selection clearly diverges from the strategy of Liquet \textit{et al.} \cite{LiquetEtAl17}. We also adapt the reasoning of Yang and Narisetty \cite{YangNarisetty20} to establish group selection consistency under some technical assumptions and an appropriate amount of sparsity. Section \ref{SecPost} is devoted to the conditional posterior distributions of the parameters in order to implement Gibbs samplers that are tested in Section \ref{SecEmp}. This empirical section is focused on a simulation study first, to evaluate and compare the efficiency of the models, then a real dataset is treated, and a short conclusion ends the paper. But, firstly, let us give some examples of what exactly we mean by `sparse', `group-sparse' and `sparse-group-sparse' settings, and let us summarize the definitions that we have chosen to retain for the well-known distributions as well as for the less usual ones, in order to avoid any misinterpretation of our results and proofs.

\begin{ex}
To explain a set of phenotypic traits, suppose that we investigate a large collection of genetic markers spread over twenty chromosomes. For coordinate sparsity (`sparse' setting), only a few markers are active. For group sparsity (`group-sparse' setting), the markers are clustered into groups (formed by chromosomes) and only a few of them are active. For sparse-group sparsity (`sparse-group-sparse' setting), only a few chromosomes are active and they are sparse, the result is a bi-level selection (chromosomes and markers). This will be the context of our example on real data (Section \ref{SecEmpReal}).
\end{ex}

\begin{defi}[Gaussian]
\label{DefN}
The density of $X \in \dR^{d_1 \times d_2}$ following the matrix normal distribution $\cMN_{d_1 \times d_2}(M,\, \Sigma_1,\, \Sigma_2)$ is given by
\begin{equation*}
p(X) = \frac{1}{(2 \pi)^{\frac{d_1\, d_2}{2}}\, \vert \Sigma_1 \vert^{\frac{d_2}{2}}\, \vert \Sigma_2 \vert^{\frac{d_1}{2}}}\, \exp\!\left(\! -\frac{1}{2}\, \tr\!\left( \Sigma_2^{-1} (X - M)^t\, \Sigma_1^{-1} (X - M) \right) \right)
\end{equation*}
where $M \in \dR^{d_1 \times d_2}$, $\Sigma_1 \in \dS_{++}^{\, d_1}$ and $\Sigma_2 \in \dS_{++}^{\, d_2}$. When $d_2=1$, this is a multivariate normal distribution $\cN_{d}(\mu,\, \Sigma)$ with $d = d_1$, $\mu = M$ and $\Sigma = \Sigma_2^{-1} \Sigma_1$, having density
\begin{equation*}
p(X) = \frac{1}{(2 \pi)^{\frac{d}{2}}\, \vert \Sigma \vert^{\frac{1}{2}}}\, \exp\!\left(\! -\frac{1}{2}\, (X - \mu)^t\, \Sigma^{-1} (X - \mu) \right)
\end{equation*}
where $\mu \in \dR^{d}$ and $\Sigma \in \dS_{++}^{\, d}$.
\end{defi}

\begin{defi}[Generalized Inverse Gaussian]
\label{DefGIG}
The density of $X \in \dS_{++}^{\, d}$ following the matrix generalized inverse Gaussian distribution $\cMGIG_d(\nu,\, A,\, B)$ is given by
\begin{equation*}
p(X) = \frac{\vert X \vert^{\nu - \frac{d+1}{2}}}{\big\vert \frac{A}{2} \big\vert^\nu\, B_\nu\big( \frac{A}{2}, \frac{B}{2} \big)}\, \exp\!\left(\! -\frac{1}{2}\, \tr\!\left( A\, X^{-1} + B\, X \right) \right)\! \ind_{\{ X\, \in\, \dS_{++}^{\, d} \}}
\end{equation*}
where $\nu \in \dR$, $A \in \dS_{++}^{\, d}$, $B \in \dS_{++}^{\, d}$ and $B_\nu$ is a Bessel-type function of order $\nu$. When $d=1$, this is a generalized inverse Gaussian distribution $\cGIG(\nu,\, a,\, b)$ with $a=A$ and $b=B$, having density
\begin{equation*}
p(X) = \frac{X^{\nu - 1}}{\big( \frac{a}{2} \big)^{\! \nu}\, B_\nu\big( \frac{a}{2}, \frac{b}{2} \big)}\, \de^{-\frac{a}{2 X} - \frac{b\, X}{2}}\, \ind_{\{ X\, >\, 0 \}}
\end{equation*}
where $\nu \in \dR$, $a > 0$ and $b > 0$.
\end{defi}

\begin{defi}[Wishart/Gamma/Exponential]
\label{DefW}
The density of $X \in \dS_{++}^{\, d}$ following the matrix Wishart distribution $\cW_d(u,\, V)$ is given by
\begin{equation*}
p(X) = \frac{\vert X \vert^{\frac{u-d-1}{2}}}{2^{\frac{d\, u}{2}}\, \Gamma_d\big( \frac{u}{2} \big)\, \vert V \vert^{\frac{u}{2}}}\, \exp\!\left(\! -\frac{1}{2}\, \tr\!\left( V^{-1} X \right) \right)\! \ind_{\{ X\, \in\, \dS_{++}^{\, d} \}}
\end{equation*}
where $u > d-1$, $V \in \dS_{++}^{\, d}$ and $\Gamma_d$ is the multivariate Gamma function of order $d$. When $d=1$, this is a Gamma distribution $\Gamma(a,\, b)$ with $a=\frac{u}{2}$ and $\frac{1}{b} = 2\, V$, having density
\begin{equation*}
p(X) = \frac{b^{\, a}\, X^{a-1}}{\Gamma(a)}\, \de^{-b\, X}\, \ind_{\{ X\, >\, 0 \}}
\end{equation*}
where $a > 0$ and $b > 0$. The exponential distribution $\cE(\ell)$ is then defined as the $\Gamma(1,\, \ell)$ distribution, for $\ell > 0$.
\end{defi}

\begin{defi}[Beta]
\label{DefB}
The density of $X \in [0,1]$ following the Beta distribution $\beta(a,\, b)$ is given by
\begin{equation*}
p(X) = \frac{X^{a-1}\, (1-X)^{\, b-1}}{\beta(a,b)}\, \ind_{\{ 0\, \leq\, X\, \leq\, 1 \}}
\end{equation*}
where $a > 0$, $b > 0$ and $\beta$ is the Beta function.
\end{defi}

In all the paper, data and parameters are gathered in $\Theta = \{ \dY, \dX, \D, \Oy, \nu, \lambda, \pi \}$ and, to standardize, for any $e \in \Theta$, we note $\Theta_{e} = \Theta \backslash \{ e \}$. 

\section{The sparse setting}
\label{SecS}

In this section, $\lambda_i \in \dR$ is the $i$-th component of $\lambda \in \dR^p$, $\D_i \in \dR^q$ is the $i$-th column of $\D$ and $\dX_i \in \dR^n$ stand for the $i$-th column of $\dX$ ($1 \leq i \leq p$). Let us consider the hierarchical Bayesian model, where the columns of $\D$ are assumed to be independent, given by
\begin{equation}
\label{ModHieraS}
\left\{
\begin{array}{lcl}
\dY\, \vert\, \dX, \D, \Oy & \sim & \cMN_{n \times q}(-\dX\, \D^t\, \Oyi,\, I_{n},\, \Oyi) \\
\D_i\, \vert\, \Oy, \lambda_i, \pi & \simindep & (1-\pi)\, \cN_q(0,\, \lambda_i\, \Oy) + \pi\, \delta_0 \\
\lambda_i & \simindep & \Gamma(\alpha,\, \ell_i) \\
\Oy & \sim & \cW_q(u,\, V) \\
\pi & \sim & \beta(a,\, b)
\end{array}
\right.
\end{equation}
for $i \in \llbracket 1, p \rrbracket$, with hyperparameters $\alpha = \frac{1}{2} (q+1)$, $\ell_i > 0$, $u > q-1$, $V \in \dS^{\, q}_{++}$, $a > 0$ and $b > 0$. A general ungrouped sparsity is promoted in the columns of $\D$ through the spike-and-slab prior. In this mixture model, $\pi$ is the prior spike probability and $\lambda$ is an adaptative shrinkage factor acting at the predictor scale ($\lambda_i$ is associated with the direct links between predictor $i$ and all the responses). When $\ell_i = \ell$ for all $i$, we will rather speak of global shrinkage. The degree of sparsity will be characterized by the number $N_0$ of zero columns of $\D$, that is
\begin{equation}
\label{N0}
N_0 = \card(i,\, \D_i = 0) = \sum_{i=1}^{p} \ind_{\{ \D_i\, =\, 0 \}}.
\end{equation}
To implement a Gibbs sampler from the full posterior distribution stemming from \eqref{ModHieraS}, we may use the conditional distributions given in the proposition below.

\begin{prop}
\label{PropPostS}
In the hierarchical model \eqref{ModHieraS}, the conditional posterior distributions are as follows.
\begin{itemize}[label=$-$, leftmargin=*]
\item The parameter $\D$ satisfies, for $i \in \llbracket 1, p \rrbracket$,
\begin{equation*}
\D_i\, \vert\, \Theta_{\D_i} ~\sim~ (1-p_i)\, \cN_q\!\left( -s_i\, H_i,\, s_i\, \Oy \right) + p_i\, \delta_0
\end{equation*}
where
\begin{equation*}
H_i = \Oy\, \dY^{\, t}\, \dX_i + \sum_{j\, \neq\, i} \langle \dX_i,\, \dX_j \rangle\, \D_j, \hsp s_i = \frac{\lambda_i}{1 + \lambda_i\, \Vert \dX_i \Vert^{\, 2}}
\end{equation*}
and
\begin{equation*}
p_i = \frac{\pi}{\pi + (1-\pi)\, (1 + \lambda_i\, \Vert \dX_i \Vert^{\, 2})^{-\frac{q}{2}}\, \exp\!\Big( \frac{s_i\, H_i^{\, t}\, \Oyi H_i}{2} \Big)}.
\end{equation*}
\item The parameter $\Oy$ satisfies
\begin{equation*}
\Oy\, \vert\, \Theta_{\Oy} ~\sim~ \cMGIG_q\!\left( \frac{n-p+N_0 + u}{2},\, \D\, (\dX^t\, \dX + D_\lambda^{-1})\, \D^t,\, \dY^{\, t}\, \dY + V^{-1} \right)
\end{equation*}
where $D_\lambda = \diag(\lambda_1, \hdots, \lambda_p)$.
\item The parameter $\lambda$ satisfies, for $i \in \llbracket 1, p \rrbracket$,
\begin{equation*}
\lambda_i\, \vert\, \Theta_{\lambda_i} ~\sim~ \ind_{\{ \D_i\, \neq\, 0 \}}\, \cGIG\!\left( \frac{1}{2},\, \D_i^t\, \Oyi \D_i,\, 2\, \ell_i \right) + \ind_{\{ \D_i\, =\, 0 \}}\, \Gamma(\alpha,\, \ell_i).
\end{equation*}
\item The parameter $\pi$ satisfies
\begin{equation*}
\pi\, \vert\, \Theta_{\pi} ~\sim~ \beta\big( N_0+a,\, p-N_0+b \big).
\end{equation*}
\end{itemize}
\end{prop}
\begin{proof}
See Section \ref{SecPostS}.
\end{proof}

\begin{rem}
\label{RemLaplace}
The Bayesian Lasso, as introduced \textit{e.g.} in Sec. 6.1 of \cite{HastieEtAl15} or in \cite{ParkCasella08}, assumes a prior Laplace distribution for the regression coefficients conditional on the noise variance. In our case, $\D_i\, \vert\, \Oy, \pi$ is still a multivariate spike-and-slab (after integrating over $\lambda_i$), with a slab following a so-called multivariate $K$-distribution (see \cite{EltoftEtAl06}), which is a generalization of the multivariate Laplace distribution. See \textit{e.g.} Sec 2.1 of \cite{LiquetEtAl17}. From this point of view, our study is in line with the usual Bayesian regression schemes. Perhaps even more interesting, going on with the idea of the authors, suppose that, for all $1 \leq i \leq p$, $\D_i = b_i\, \D_i^{*}$ where $\D_i^{*}$ follows the multivariate $K$-distribution described above and $b_i\, \vert\, \pi \sim \cB(1-\pi)$ is independent of $\D_i^{*}$. Now, the sparsity in $\D$ is not induced by a spike-and-slab strategy anymore but, equivalently, by multiplying the slab part by an independent Bernoulli variable being 0 with probability $\pi$. Then, it is possible to show that the negative log-likelihood of this alternative hierarchical model is given, up to an additive constant that does not depend on $\D$, by
\begin{equation*}
\frac{1}{2}\, \left\Vert (\dY + \dX\, \D^t\, \Oyi)\, \Oy^{\frac{1}{2}} \right\Vert_{F}^2+ \sum_{i=1}^p c_i \left\Vert \Oy^{-\frac{1}{2}} \D_i^{*} \right\Vert_{F} + \ln\!\left( \frac{1-\pi}{\pi} \right) \sum_{i=1}^p b_i
\end{equation*}
where $c_i > 0$. We first recognize an $\ell_2$-type penalty but also an $\ell_0$-type penalty on $\D$ (provided that $\pi < \frac{1}{2}$) since summing the $b_i$ amounts to counting the number of non-zero columns in $\D$. Consequently, there is a close connection between our hierarchical Bayesian model and the regressions penalized by $\ell_2$ and $\ell_0$ norms, problems that are known to be very hard to solve due to combinatorial optimization.
\end{rem}

The particular case $q=1$ is a very useful corollary of the proposition. Here, the direct links form a row vector such that $\D^t \in \dR^p$ with components $\D_i \in \dR$ ($1 \leq i \leq p$), and the precision matrix of the responses reduces to $\oy > 0$. According to the parametrization of the distributions (see Section \ref{SecIntro}), the corresponding prior distribution of $\oy$ is $\Gamma(\frac{u}{2}, \frac{1}{2\, v})$ for $u, v > 0$ and the one of $\lambda_i$ is $\cE(\ell_i)$ for $\ell_i > 0$. The other priors are unchanged.

\begin{cor}
\label{CorPostS1}
In the hierarchical model \eqref{ModHieraS} with $q=1$, the conditional posterior distributions are as follows.
\begin{itemize}[label=$-$, leftmargin=*]
\item The parameter $\D$ satisfies, for $i \in \llbracket 1, p \rrbracket$,
\begin{equation*}
\D_i\, \vert\, \Theta_{\D_i} ~\sim~ (1-p_i)\, \cN\!\left( -s_i\, h_i,\, s_i\, \oy \right) + p_i\, \delta_0
\end{equation*}
where
\begin{equation*}
h_i = \oy\, \langle \dX_i,\, \dY \rangle + \sum_{j\, \neq\, i} \langle \dX_i,\, \dX_j \rangle\, \D_j, \hsp s_i = \frac{\lambda_i}{1 + \lambda_i\, \Vert \dX_i \Vert^{\, 2}}
\end{equation*}
and
\begin{equation*}
p_i = \frac{\pi}{\pi + (1-\pi)\, (1 + \lambda_i\, \Vert \dX_i \Vert^{\, 2})^{-\frac{1}{2}}\, \exp\!\Big( \frac{s_i\, h_i^2}{2\, \oy} \Big)}.
\end{equation*}
\item The parameter $\oy$ satisfies
\begin{equation*}
\oy\, \vert\, \Theta_{\oy} ~\sim~ \cGIG\!\left( \frac{n-p+N_0+u}{2},\, \D\, (\dX^t\, \dX + D_\lambda^{-1})\, \D^t,\, \Vert \dY \Vert^{\, 2} + \frac{1}{v} \right)
\end{equation*}
where $D_\lambda = \diag(\lambda_1, \hdots, \lambda_p)$.
\item The parameter $\lambda$ satisfies, for $i \in \llbracket 1, p \rrbracket$,
\begin{equation*}
\lambda_i\, \vert\, \Theta_{\lambda_i} ~\sim~ \ind_{\{ \D_i\, \neq\, 0 \}}\, \cGIG\!\left( \frac{1}{2},\, \frac{\D_i^{\, 2}}{\oy},\, 2\, \ell_i \right) + \ind_{\{ \D_i\, =\, 0 \}}\, \cE(\ell_i).
\end{equation*}
\item The parameter $\pi$ satisfies
\begin{equation*}
\pi\, \vert\, \Theta_{\pi} ~\sim~ \beta\big( N_0+a,\, p-N_0+b \big).
\end{equation*}
\end{itemize}
\end{cor}
\begin{proof}
This is a consequence of Proposition \ref{PropPostS}.
\end{proof}

Note that we can also easily derive the Bayesian counterpart of the standard PGGM adapted to the small-dimensional case, with no sparsity, by taking $\pi=0$.

\begin{cor}
\label{CorPostNS}
In the hierarchical model \eqref{ModHieraS} with $\pi=0$, the conditional posterior distributions are as follows.
\begin{itemize}[label=$-$, leftmargin=*]
\item The parameter $\D$ satisfies, for $i \in \llbracket 1, p \rrbracket$,
\begin{equation*}
\D_i\, \vert\, \Theta_{\D_i} ~\sim~ \cN_q\!\left( -s_i\, H_i,\, s_i\, \Oy \right)
\end{equation*}
where
\begin{equation*}
H_i = \Oy\, \dY^{\, t}\, \dX_i + \sum_{j\, \neq\, i} \langle \dX_i,\, \dX_j \rangle\, \D_j \hsp \text{and} \hsp s_i = \frac{\lambda_i}{1 + \lambda_i\, \Vert \dX_i \Vert^{\, 2}}.
\end{equation*}
\item The parameter $\Oy$ satisfies
\begin{equation*}
\Oy\, \vert\, \Theta_{\Oy} ~\sim~ \cMGIG_q\!\left( \frac{n-p + u}{2},\, \D\, (\dX^t\, \dX + D_\lambda^{-1})\, \D^t,\, \dY^{\, t}\, \dY + V^{-1} \right)
\end{equation*}
where $D_\lambda = \diag(\lambda_1, \hdots, \lambda_p)$.
\item The parameter $\lambda$ satisfies, for $i \in \llbracket 1, p \rrbracket$,
\begin{equation*}
\lambda_i\, \vert\, \Theta_{\lambda_i} ~\sim~ \cGIG\!\left( \frac{1}{2},\, \D_i^t\, \Oyi \D_i,\, 2\, \ell_i \right).
\end{equation*}
\end{itemize}
\end{cor}
\begin{proof}
This is a consequence of Proposition \ref{PropPostS}.
\end{proof}

In the simulation study of Section \ref{SecEmpSim}, Scen. 0, 1 and 2 are dedicated to the sparse setting. The next section discusses the group sparsity in $\D$.

\section{The group-sparse setting}
\label{SecGS}

The predictors are now ordered in $m$ groups of sizes $\kappa_1 + \hdots + \kappa_m = p$. For the $g$-th group ($1 \leq g \leq m$), $\lambda_g \in \dR$ is the $g$-th component of $\lambda \in \dR^m$, the covariate submatrix is $\gdX_g \in \dR^{n \times \kappa_g}$ and the corresponding slice of $\D$ is $\gD_g \in \dR^{q \times \kappa_g}$. Let us consider the hierarchical Bayesian model, where the columns of $\D$ are assumed to be independent both within and between the groups, given by
\begin{equation}
\label{ModHieraGS}
\left\{
\begin{array}{lcl}
\dY\, \vert\, \dX, \D, \Oy & \sim & \cMN_{n \times q}(-\dX\, \D^t\, \Oyi,\, I_{n},\, \Oyi) \\
\gD_g\, \vert\, \Oy, \lambda_g, \pi & \simindep & (1-\pi)\, \cMN_{q \times \kappa_g}(0,\, \lambda_g\, \Oy, I_{\kappa_g}) + \pi\, \delta_0 \\
\lambda_g & \simindep & \Gamma(\alpha_g,\, \ell_g) \\
\Oy & \sim & \cW_q(u,\, V) \\
\pi & \sim & \beta(a,\, b)
\end{array}
\right.
\end{equation}
for $g \in \llbracket 1, m \rrbracket$, with hyperparameters $\alpha_g = \frac{1}{2} (q\, \kappa_g +1)$, $\ell_g > 0$, $u > q-1$, $V \in \dS^{\, q}_{++}$, $a > 0$ and $b > 0$. A general group sparsity is promoted in the columns of $\D$ through the spike-and-slab prior at the group level. In this mixture model, $\pi$ is the prior spike probability and $\lambda$ is an adaptative shrinkage factor acting at the group scale ($\lambda_g$ is associated with the direct links between the predictors of group $g$ and all the responses). Likewise, when $\ell_g = \ell$ for all $g$, we will rather speak of global shrinkage. Now, the degree of sparsity will be characterized by $N_0$ given in \eqref{N0}, but also by the number $G_0$ of zero groups of $\D$, that is
\begin{equation}
\label{G0}
G_0 = \card(g,\, \gD_g = 0) = \sum_{g=1}^{m} \ind_{\{ \D_g\, =\, 0 \}}.
\end{equation}
To implement a Gibbs sampler from the full posterior distribution stemming from \eqref{ModHieraGS}, we may use the conditional distributions given in the proposition below.

\begin{prop}
\label{PropPostGS}
In the hierarchical model \eqref{ModHieraGS}, the conditional posterior distributions are as follows.
\begin{itemize}[label=$-$, leftmargin=*]
\item The parameter $\D$ satisfies, for $g \in \llbracket 1, m \rrbracket$,
\begin{equation*}
\gD_g\, \vert\, \Theta_{\D_g} ~\sim~ (1-p_g)\, \cMN_{q \times \kappa_g}\!\left( -H_g\, S_g,\, \Oy,\, S_g \right) + p_g\, \delta_0
\end{equation*}
where
\begin{equation*}
H_g = \Oy \dY^{\, t}\, \gdX_g + \sum_{j\, \neq\, g} \gD_j\, \gdX_j^{\, t}\, \gdX_g, \hsp S_g = \lambda_g\, \big( I_{\kappa_g} + \lambda_g\, \gdX_g^{\, t}\, \gdX_g \big)^{-1}
\end{equation*}
and
\begin{equation*}
p_g = \frac{\pi}{\pi + (1-\pi)\, \vert I_{\kappa_g} + \lambda_g\, \gdX_g^{\, t}\, \gdX_g \vert^{-\frac{q}{2}}\, \exp\!\Big( \frac{\tr(H_g^{\, t}\, \Oyi H_g\, S_g)}{2} \Big)}.
\end{equation*}
\item The parameter $\Oy$ satisfies
\begin{equation*}
\Oy\, \vert\, \Theta_{\Oy} ~\sim~ \cMGIG_q\!\left( \frac{n-p+N_0 + u}{2},\, \D\, (\dX^t\, \dX + D_\lambda^{-1})\, \D^t,\, \dY^{\, t}\, \dY + V^{-1} \right)
\end{equation*}
where $D_\lambda = \diag(\lambda_1, \hdots, \lambda_1, \hdots, \lambda_m, \hdots, \lambda_m)$ with each $\lambda_g$ duplicated $\kappa_g$ times.
\item The parameter $\lambda$ satisfies, for $g \in \llbracket 1, m \rrbracket$,
\begin{equation*}
\lambda_g\, \vert\, \Theta_{\lambda_g} ~\sim~ \ind_{\{ \D_g\, \neq\, 0 \}}\, \cGIG\!\left( \frac{1}{2},\, \tr(\gD_g^t\, \Oyi \gD_g),\, 2\, \ell_g \right) + \ind_{\{ \D_g\, =\, 0 \}}\, \Gamma( \alpha_g,\, \ell_g).
\end{equation*}
\item The parameter $\pi$ satisfies
\begin{equation*}
\pi\, \vert\, \Theta_{\pi} ~\sim~ \beta\big( G_0+a,\, m-G_0+b \big).
\end{equation*}
\end{itemize}
\end{prop}
\begin{proof}
See Section \ref{SecPostGS}.
\end{proof}

Note that Remark \ref{RemLaplace} still applies to this configuration, after some adjustments (the $\ell_0$-like penalty is on the number of non-zero groups). Here again, the particular case $q=1$ is a very useful corollary. The direct links form a row vector such that $\D^t \in \dR^p$ with groups $\gD_g^t \in \dR^{\kappa_g}$ ($1 \leq g \leq m$), the precision matrix of the responses reduces to $\oy > 0$. According to the parametrization of the distributions (see Section \ref{SecIntro}), the corresponding prior distribution of $\oy$ is $\Gamma(\frac{u}{2}, \frac{1}{2\, v})$ for $u, v > 0$, like in the ungrouped setting. The other priors are unchanged.

\begin{cor}
\label{CorPostGS1}
In the hierarchical model \eqref{ModHieraGS} with $q=1$, the conditional posterior distributions are as follows.
\begin{itemize}[label=$-$, leftmargin=*]
\item The parameter $\D$ satisfies, for $g \in \llbracket 1, m \rrbracket$,
\begin{equation*}
\gD_g^t\, \vert\, \Theta_{\D_g} ~\sim~ (1-p_g)\, \cN_{\kappa_g}\!\left( -S_g\, H_g,\, \oy\, S_g \right) + p_g\, \delta_0
\end{equation*}
where
\begin{equation*}
H_g = \oy\, \gdX_g^{\, t}\, \dY + \sum_{j\, \neq\, g} \gdX_g^{\, t}\, \gdX_j\, \gD_j^t, \hsp S_g = \lambda_g\, \big( I_{\kappa_g} + \lambda_g\, \gdX_g^{\, t}\, \gdX_g \big)^{-1}
\end{equation*}
and
\begin{equation*}
p_g = \frac{\pi}{\pi + (1-\pi)\, \vert I_{\kappa_g} + \lambda_g\, \gdX_g^{\, t}\, \gdX_g \vert^{-\frac{1}{2}}\, \exp\!\Big( \frac{H_g^{\, t}\, S_g\, H_g}{2\, \oy} \Big)}.
\end{equation*}
\item The parameter $\oy$ satisfies
\begin{equation*}
\oy\, \vert\, \Theta_{\oy} ~\sim~ \cGIG\!\left( \frac{n-p+N_0+u}{2},\, \D\, (\dX^t\, \dX + D_\lambda^{-1})\, \D^t,\, \Vert \dY \Vert^{\, 2} + \frac{1}{v} \right)
\end{equation*}
where $D_\lambda = \diag(\lambda_1, \hdots, \lambda_1, \hdots, \lambda_m, \hdots, \lambda_m)$ with each $\lambda_g$ duplicated $\kappa_g$ times.
\item The parameter $\lambda$ satisfies, for $g \in \llbracket 1, m \rrbracket$,
\begin{equation*}
\lambda_g\, \vert\, \Theta_{\lambda_g} ~\sim~ \ind_{\{ \D_g\, \neq\, 0 \}}\, \cGIG\!\left( \frac{1}{2},\, \frac{\Vert \gD_g \Vert^{\, 2}}{\oy},\, 2\, \ell_g \right) + \ind_{\{ \D_g\, =\, 0 \}}\, \Gamma(\alpha_g,\, \ell_g).
\end{equation*}
\item The parameter $\pi$ satisfies
\begin{equation*}
\pi\, \vert\, \Theta_{\pi} ~\sim~ \beta\big( G_0+a,\, m-G_0+b \big).
\end{equation*}
\end{itemize}
\end{cor}
\begin{proof}
This is a consequence of Proposition \ref{PropPostGS}.
\end{proof}

In the simulation study of Section \ref{SecEmpSim}, Scen. 3 and 4 are dedicated to the group-sparse setting. To conclude this section, a theoretical guarantee is provided (given $\Oy$ and with $\lambda = \lambda_n$ and $\pi = \pi_n$ depending on $n$). It is possible to obtain a model selection consistency property for this approach when both the number of observations $n$ and the number of groups $m = m_n$ tend to infinity, by adapting the reasoning of \cite{YangNarisetty20} dedicated to the linear regression (with $q=1$). Indeed, when $\Oy$ is known, $\D$ reduces to a linear transformation of $B$. Thus, it is not surprising that a similar result follows under the same kind of hypotheses. In the sequel, we denote by $\dX_{(k)} \in \dR^{n \times \vert k \vert}$ the design matrix of rank $r_k$ corresponding to the submodel indexed by the binary vector $k \in \{ 0,1 \}^m$ having $\vert k \vert$ non-zero values ($k_g=1$ means that the $g$-th group is included in the model), and by $\Pi_{(k)} \in \dR^{n \times n}$ the projection matrix onto the column-space of $\dX_{(k)}$. Similarly, $\D$ restricted to $k$ is $\D_{(k)} \in \dR^{q \times \vert k \vert}$. The true model is called $t$ and $t^{\pm g}$ are submodels of $t$ that contain only the $g$-th group or that are deprived of it, respectively. Let
\begin{equation*}
\delta_1 = \inf_{1\, \leq\, g\, \leq\, \vert t \vert}\, \big\Vert (I_n - \Pi_{(t^{-g})})\, \dX_{(t^{+g})}\, \D_{(t^{+g})}^t\, \Oy^{-\frac{1}{2}} \big\Vert_F^2
\end{equation*}
and, for some $K > 0$,
\begin{equation*}
\delta_2^K = \inf_{k\, \in\, E_K}\, \big\Vert (I_n - \Pi_{(k)})\, \dX_{(t)}\, \D_{(t)}^t\, \Oy^{-\frac{1}{2}} \big\Vert_F^2
\end{equation*}
with $E_K = \{ k \text{ such that } t \not\subset k \text{ and } r_k \leq K r_t \}$. Let also, 
\begin{equation*}
\mu_{n,\, \text{min}}^K = \inf_{k\, \in\, F_K} ~ \mu^+\bigg( \frac{\dX_{(k)}^t\, \dX_{(k)}}{n} \bigg) \hsp \text{and} \hsp \bar{\mu}_n = \inf_{k\, \in\, F} ~ \mu^*\bigg( \frac{\dX_{(k)}^t\, (I_n - \Pi_{(k\, \cap\, t)})\, \dX_{(k)}}{n} \bigg)
\end{equation*}
with $F_K = \{ k \text{ such that } t \subset k \text{ and } r_k \leq (K+1)\, r_t \}$ and $F = \{ k \text{ such that } \vert k \backslash t \vert > 0 \}$, and where, for a square matrix $A$, $\mu^+(A)$ is the minimum non-zero eigenvalue of $A$ and $\mu^*(A)$ is the geometric mean of the non-zero eigenvalues of $A$. The hypotheses are those of \cite{YangNarisetty20} that we have to slightly adapt. By $f_n \asymp g_n$ we mean that there is a constant $c \neq 0$ such that $f_n/g_n \rightarrow c$ as $n$ tends to infinity.

\begin{enumerate}[label={(H.\arabic*)}]
\item \label{Hyp1} There exists a rate such that $m_n = \de^{v_n}$ with $v_n \rightarrow +\infty$ and $v_n = o(n)$. 
\item \label{Hyp2} The prior slab probability satisfies $1-\pi_n \asymp 1/m_n$.
\item \label{Hyp3} The shrinkage factors satisfy $n \lambda_n^{\sharp} \asymp m_n^{2+\eta}\, \bar{\mu}_n^{-\eta}$ and $\mu_{n,\, \text{min}}^K\, n \lambda_n^{\sharp} \rightarrow +\infty$ for some $\eta > 0$, where $\lambda_n^{\sharp} = \max_i \lambda_{n,\, i}$.
\item \label{Hyp4} There exists $\epsilon_1 > 0$ such that $\delta_1 > (1+\epsilon_1)\, r_t\, [(4+\eta) \ln m_n - \eta\, \ln \bar{\mu}_n]$.
\item \label{Hyp5} There exists $\epsilon_2 > 0$ such that $\delta_2^K > (1+\epsilon_2)\, r_t\, [(4+\eta) \ln m_n - \eta\, \ln \bar{\mu}_n]$ for some $K > \max(8/\eta+1, \eta/(\eta-1))$.
\end{enumerate}
We refer the reader to p. 917 of \cite{YangNarisetty20} where the authors give very clarifying comments on the interpretation to be given to these technical assumptions. In particular, while \ref{Hyp1}, \ref{Hyp2} and \ref{Hyp3} control the behavior of $m_n$, $\pi_n$ and $\lambda_n$ as $n$ tends to infinity, \ref{Hyp4} and \ref{Hyp5} are related to sensitivity and specificity and are therefore in connection with the true model $t$.

\begin{prop}
\label{PropSupRecGS}
Suppose that  \ref{Hyp1}--\ref{Hyp5} are satisfied. Then, as $n$ tends to infinity,
\begin{equation*}
\dP(\cT\, \vert\, \dY, \dX, \Oy) \limp 1
\end{equation*}
where $\cT = \{ t \text{ is selected} \}$.
\end{prop}
\begin{proof}
The result is obtained by following the same lines as the proof of Thm 2.1 of \cite{YangNarisetty20}. One just has to clarify a few points to solve the issues arising from $q \geq 1$ and from the adaptative shrinkage, which is done in Section \ref{SecSupRecGS}.
\end{proof}

\begin{rem}
\label{RemSparsity}
Obviously, Proposition \ref{PropSupRecGS} also holds for the sparse setting (with $m=p$) and in that case, it is instructive to draw the parallel with Thm. 1 of \cite{RenEtAl15} even if the estimation procedure is very different. The authors show that a necessary and sufficient condition to obtain the $\sqrt{n}$-consistent estimation of the precision matrix in a GGM is the presence of at most $\asymp \sqrt{n}/\ln p$ non-zero columns. In the Gibbs sampler (see Proposition \ref{PropPostS}), the slab probability $1-\pi$ is generated according to a distribution that satisfies
\begin{equation*}
\dE[1-\pi\, \vert\, \Theta_{\pi}] = \frac{p-N_0+b}{p+a+b} \hsp \text{and} \hsp \dV(1-\pi\, \vert\, \Theta_{\pi}) = \frac{(N_0+a)(p-N_0+b)}{(p+a+b)^2\, (p+a+b+1)}.
\end{equation*}
Thus, if the model selects $\asymp \sqrt{n}/\ln p$ predictors, it follows that the posterior expectation of $1-\pi$ is $\asymp \sqrt{n}/(p \ln p) = 1/p$ when $p = \de^{\sqrt{n}}$. In that case, the posterior variance of $1-\pi$ is $\asymp 1/p^2$. To sum up, in a model with $\asymp \sqrt{n}/\ln p$ predictors selected, the posterior distribution of $1-\pi$ is very concentrated around $1/p$ which conforms to \ref{Hyp1} and \ref{Hyp2}. This is not directly comparable due to the different procedures, but it seems interesting to observe that the same orders of magnitude are involved to reach theoretical guarantees for the estimation of $\Delta$.
\end{rem}

In the next section, an approach is suggested to deal with sparse-group sparsity in $\D$, for a bi-level selection.

\section{The sparse-group-sparse setting}
\label{SecSGS}

To produce a sparse model both at the variable level (for variable selection) and at the group level (for group selection), it seems natural to carry on with our strategy by introducing another spike-and-slab effect into the first one. The predictors are still ordered in $m$ groups of sizes $\kappa_1 + \hdots + \kappa_m = p$. For the $g$-th group ($1 \leq g \leq m$), $\lambda_g \in \dR$ is the $g$-th component of $\lambda \in \dR^m$ and, for the $i$-th predictor of this group ($1 \leq i \leq \kappa_g$), $\nu_{gi} \in \dR$ is the $i$-th component of $\nu_g \in \dR^{\kappa_g}$. The $i$-th column of the covariate submatrix $\gdX_g$ is $\dX_{gi} \in \dR^n$ and the corresponding slice of $\gD_g$ is $\D_{gi} \in \dR^q$ while $\gD_{g \backslash i} \in \dR^{q \times (\kappa_g-1)}$ is $\gD_g$ deprived of $\D_{gi}$. Here our approach diverges from \cite{XuGosh15} and \cite{LiquetEtAl17}. The bi-level selection of the authors is made through spike-and-slab effects both at the group scale and on the individual variances, considered as truncated Gaussians, generating zero groups and (almost surely) zero coefficients within the groups. Let us suggest instead the Bayesian hierarchical model given by
\begin{equation}
\label{ModHieraSGS}
\left\{
\begin{array}{lcl}
\dY\, \vert\, \dX, \D, \Oy & \sim & \cMN_{n \times q}(-\dX\, \D^t\, \Oyi,\, I_{n},\, \Oyi) \\
\gD_g\, \vert\, \nu_g, \lambda_g, \pi & \simindep & (1-\pi_1)\, \big[ (1-\pi_2)\, \cN_q(0,\, \lambda_g\, \nu_{gi}\, \Oy) + \pi_2\, \delta_0 \big]^{\otimes\, \kappa_g} + \pi_1\, \delta_0 \\
\nu_{gi} & \simindep & \Gamma(\alpha,\, \ell_{gi}) \\
\lambda_g & \simindep & \Gamma(\alpha_g,\, \gamma_g) \\
\Oy & \sim & \cW_q(u,\, V) \\
\pi_j & \simindep & \beta(a_j,\, b_j)
\end{array}
\right.
\end{equation}
for $g \in \llbracket 1, m \rrbracket$, $i \in \llbracket 1, \kappa_g \rrbracket$ and $j \in \llbracket 1, 2 \rrbracket$, with hyperparameters $\alpha = \frac{1}{2} (q +1)$, $\alpha_g = \frac{1}{2} (q\, \kappa_g +1)$, $\ell_{gi} > 0$, $\gamma_g > 0$, $u > q-1$, $V \in \dS^{\, q}_{++}$, $a_j > 0$, and $b_j > 0$. In this mixture model, $\pi_1$ is the prior spike probability on the groups whereas $\pi_2$ is the prior spike probability within the non-zero groups, for a bi-level selection. In terms of cumulative shrinkage effects, $\lambda$ is an adaptative shrinkage factor acting at the group scale and $\nu$ is an adaptative shrinkage factor acting at the predictor scale ($\lambda_g$ is associated with the direct links between the predictors of group $g$ and all the responses whereas $\nu_{gi}$ is associated with the direct links between predictor $i$ of group $g$ and all the responses). In this way, \eqref{ModHieraSGS} opens up many perspectives for dealing with bi-level shrinkage. We can set $\gamma_g = \gamma$ for all $g$, for a global shrinkage at the group scale. At the predictor scale, when $\ell_{gi} = \ell_g$ for all $i$, this is a global shrinkage in the $g$-th group but we might even consider a full global shrinkage $\ell_{gi} = \ell$. However, an identifiability issue may result from the product $\lambda_g\, \nu_{gi}$ between group and within-group effects. Even if the posterior distributions depend on different levels of data that shall resolve it, one can for example fix $\lambda_g=1$ (for adaptative) or $\nu_{gi}=1$ (for global) and let the shrinkage entirely rely on the other parameter. Although it achieves the same objectives as those of \cite{XuGosh15} and \cite{LiquetEtAl17}, this hierarchy seems more consistent with our previous sections (take $\pi_2 = 0$ and $\nu_{gi} = 1$ to remove the within-group effect and recover the group-sparse setting of Section \ref{SecGS}, take $\pi_1 = 0$ and $\lambda_{g} = 1$ to remove the group effect and recover the sparse setting of Section \ref{SecS}). In this context, the degree of sparsity is still characterized by $N_0$ given in \eqref{N0} for the predictor scale, by $G_0$ given in \eqref{G0} for the group scale, but also, for the within-group scale, by the number $N_{0g}$ of zero columns in each particular group $g$, that is, for all $1 \leq g \leq m$,
\begin{equation}
\label{N0g}
N_{0g} = \card(i,\, \D_{gi} = 0) = \sum_{i=1}^{\kappa_g} \ind_{\{ \D_{gi}\, =\, 0 \}}.
\end{equation}
We also need to define the number $J_0$ of zero columns in the non-zero groups, that is
\begin{equation}
\label{J0}
J_0 = \card(i,\, \D_{gi} = 0 ~\text{and}~ \gD_g \neq 0) = \sum_{g=1}^{m} N_{0g}\, \ind_{\{ \D_g\, \neq\, 0 \}}.
\end{equation}
To implement a Gibbs sampler from the full posterior distribution stemming from \eqref{ModHieraSGS}, we may use the conditional distributions given in the proposition below.

\begin{prop}
\label{PropPostSGS}
In the hierarchical model \eqref{ModHieraSGS}, the conditional posterior distributions are as follows.
\begin{itemize}[label=$-$, leftmargin=*]
\item The parameter $\D_{gi}$ satisfies, for $g \in \llbracket 1, m \rrbracket$ and $i \in \llbracket 1, \kappa_g \rrbracket$,
\begin{equation*}
\D_{gi}\, \vert\, \Theta_{\D_{gi}} ~\sim~ (1-p_{gi})\, \cN_q\!\left( -s_{gi}\, H_{gi},\, s_{gi}\, \Oy \right) + p_{gi}\, \delta_0
\end{equation*}
where
\begin{equation*}
H_{gi} = \Oy\, \dY^{\, t}\, \dX_{gi} + \sum_{h, j\, \neq\, g, i} \langle \dX_{gi},\, \dX_{hj} \rangle\, \D_{hj}, \hsp s_{gi} = \frac{\nu_{gi}\, \lambda_g}{1 + \nu_{gi}\, \lambda_g\, \Vert \dX_{gi} \Vert^{\, 2}}
\end{equation*}
and
\begin{equation*}
p_{gi} = \frac{\rho_{gi}}{\rho_{gi} + (1-\pi_1)\, (1-\pi_2)\, (1 + \nu_{gi}\, \lambda_g\, \Vert \dX_{gi} \Vert^{\, 2})^{-\frac{q}{2}}\, \exp\!\Big( \frac{s_{gi}\, H_{gi}^{\, t}\, \Oyi\, H_{gi}}{2} \Big)}
\end{equation*}
in which $\rho_{gi} = (1-\pi_1)\, \pi_2\, \ind_{\{ \D_{g \backslash i}\, \neq\, 0 \}}\, + \pi_1\, \ind_{\{ \D_{g \backslash i}\, =\, 0 \}}$.
\item The parameter $\Oy$ satisfies
\begin{equation*}
\Oy\, \vert\, \Theta_{\Oy} ~\sim~ \cMGIG_q\!\left( \frac{n-p+N_0 + u}{2},\, \D\, (\dX^t\, \dX + D_{\lambda \nu}^{-1})\, \D^t,\, \dY^{\, t}\, \dY + V^{-1} \right)
\end{equation*}
where $D_{\lambda \nu} = \diag(\nu_{11}\, \lambda_1, \hdots, \nu_{1\kappa_1}\, \lambda_1, \hdots, \nu_{m1}\, \lambda_m, \hdots, \nu_{m\kappa_m}\, \lambda_m)$.
\item The parameter $\nu$ satisfies, for $g \in \llbracket 1, m \rrbracket$ and $i \in \llbracket 1, \kappa_g \rrbracket$,
\begin{equation*}
\nu_{gi}\, \vert\, \Theta_{\nu_{gi}} ~\sim~ \ind_{\{ \D_{gi}\, \neq\, 0 \}}\, \cGIG\!\left( \frac{1}{2},\, \frac{\D_{gi}^t\, \Oyi \D_{gi}}{\lambda_g},\, 2\, \ell_{gi} \right) + \ind_{\{ \D_{gi}\, =\, 0 \}}\, \Gamma(\alpha,\, \ell_{gi}).
\end{equation*}
\item The parameter $\lambda$ satisfies, for $g \in \llbracket 1, m \rrbracket$,
\begin{equation*}
\lambda_g\, \vert\, \Theta_{\lambda_g} ~\sim~ \ind_{\{ \D_g\, \neq\, 0 \}}\, \cGIG\!\left( \frac{q N_{0g} + 1}{2},\, \tr(D_{\nu_g}^{-1} \gD_g^t\, \Oyi \gD_g),\, 2\, \gamma_g \right) + \ind_{\{ \D_{g}\, =\, 0 \}}\, \Gamma(\alpha_g,\, \gamma_g)
\end{equation*}
where $D_{\nu_g} = \diag(\nu_{g1}, \hdots, \nu_{g\kappa_g})$.
\item The parameter $\pi$ satisfies, for $j \in \llbracket 1, 2 \rrbracket$,
\begin{equation*}
\pi_j\, \vert\, \Theta_{\pi_j} ~\sim~ \beta\big( A_j+a_j,\, B_j+b_j \big).
\end{equation*}
where $A_1 = G_0$, $B_1 = m-G_0$, $A_2 = J_0$ and $B_2 = p-N_0$.
\end{itemize}
\end{prop}
\begin{proof}
See Section \ref{SecPostSGS}.
\end{proof}

It only remains to give the explicit results for the particular case $q=1$. The direct links form a row vector such that $\D^t \in \dR^p$ with groups $\gD_g^t \in \dR^{\kappa_g}$ ($1 \leq g \leq m$) containing predictors $\D_{gi} \in \dR$ ($1 \leq i \leq \kappa_g$), and the precision matrix of the responses reduces to $\oy > 0$. According to the parametrization of the distributions (see Section \ref{SecIntro}), the corresponding prior distribution of $\oy$ is $\Gamma(\frac{u}{2}, \frac{1}{2\, v})$ for $u, v > 0$, like in the other settings, and the one of $\nu_{gi}$ is $\cE(\ell_{gi})$ for $\ell_{gi} > 0$. The other priors are unchanged.

\begin{cor}
\label{CorPostSGS1}
In the hierarchical model \eqref{ModHieraSGS} with $q=1$, the conditional posterior distributions are as follows.
\begin{itemize}[label=$-$, leftmargin=*]
\item The parameter $\D_{gi}$ satisfies, for $g \in \llbracket 1, m \rrbracket$ and $i \in \llbracket 1, \kappa_g \rrbracket$,
\begin{equation*}
\D_{gi}\, \vert\, \Theta_{\D_{gi}} ~\sim~ (1-p_{gi})\, \cN\!\left( -s_{gi}\, h_{gi},\, s_{gi}\, \oy \right) + p_{gi}\, \delta_0
\end{equation*}
where
\begin{equation*}
h_{gi} = \oy\, \langle \dX_{gi},\, \dY \rangle + \sum_{h, j\, \neq\, g, i} \langle \dX_{gi},\, \dX_{hj} \rangle\, \D_{hj}, \hsp s_{gi} = \frac{\nu_{gi}\, \lambda_g}{1 + \nu_{gi}\, \lambda_g\, \Vert \dX_{gi} \Vert^{\, 2}}
\end{equation*}
and
\begin{equation*}
p_{gi} = \frac{\rho_{gi}}{\rho_{gi} + (1-\pi_1)\, (1-\pi_2)\, (1 + \nu_{gi}\, \lambda_g\, \Vert \dX_{gi} \Vert^{\, 2})^{-\frac{1}{2}}\, \exp\!\Big( \frac{s_{gi}\, h_{gi}^2}{2\, \oy} \Big)}
\end{equation*}
in which $\rho_{gi} = (1-\pi_1)\, \pi_2\, \ind_{\{ \D_{g \backslash i}\, \neq\, 0 \}}\, + \pi_1\, \ind_{\{ \D_{g \backslash i}\, =\, 0 \}}$.
\item The parameter $\oy$ satisfies
\begin{equation*}
\oy\, \vert\, \Theta_{\oy} ~\sim~ \cGIG\!\left( \frac{n-p+N_0 + u}{2},\, \D\, (\dX^t\, \dX + D_{\lambda \nu}^{-1})\, \D^t,\, \dY^{\, t}\, \dY + \frac{1}{v} \right)
\end{equation*}
where $D_{\lambda \nu} = \diag(\nu_{11}\, \lambda_1, \hdots, \nu_{1\kappa_1}\, \lambda_1, \hdots, \nu_{m1}\, \lambda_m, \hdots, \nu_{m\kappa_m}\, \lambda_m)$.
\item The parameter $\nu$ satisfies, for $g \in \llbracket 1, m \rrbracket$ and $i \in \llbracket 1, \kappa_g \rrbracket$,
\begin{equation*}
\nu_{gi}\, \vert\, \Theta_{\nu_{gi}} ~\sim~ \ind_{\{ \D_{gi}\, \neq\, 0 \}}\, \cGIG\!\left( \frac{1}{2},\, \frac{\D_{gi}^2}{\lambda_g\, \oy},\, 2\, \ell_{gi} \right) + \ind_{\{ \D_{gi}\, =\, 0 \}}\, \cE(\ell_{gi}).
\end{equation*}
\item The parameter $\lambda$ satisfies, for $g \in \llbracket 1, m \rrbracket$,
\begin{equation*}
\lambda_g\, \vert\, \Theta_{\lambda_g} ~\sim~ \ind_{\{ \D_g\, \neq\, 0 \}}\, \cGIG\!\left( \frac{N_{0g} + 1}{2},\, \frac{\gD_g\, D_{\nu_g}^{-1}\, \gD_g^t}{\oy},\, 2\, \gamma_g \right) + \ind_{\{ \D_{g}\, =\, 0 \}}\, \Gamma(\alpha_g,\, \gamma_g)
\end{equation*}
where $D_{\nu_g} = \diag(\nu_{g1}, \hdots, \nu_{g\kappa_g})$.
\item The parameter $\pi$ satisfies, for $j \in \llbracket 1, 2 \rrbracket$,
\begin{equation*}
\pi_j\, \vert\, \Theta_{\pi_j} ~\sim~ \beta\big( A_j+a_j,\, B_j+b_j \big).
\end{equation*}
where $A_1 = G_0$, $B_1 = m-G_0$, $A_2 = J_0$ and $B_2 = p-N_0$.
\end{itemize}
\end{cor}
\begin{proof}
This is a consequence of Proposition \ref{PropPostSGS}.
\end{proof}

In the simulation study of Section \ref{SecEmpSim}, Scen. 5 and 6 are dedicated to the sparse-group-sparse setting. Now, let us prove our assertions by a few computational steps.

\section{Conditional posterior distributions}
\label{SecPost}

\subsection{The sparse setting: proof of Proposition \ref{PropPostS}}
\label{SecPostS}

First of all, the full posterior distribution of the parameters conditional on $\dX$ and $\dY$ satisfies
\begin{eqnarray}
\label{FullS}
p(\D, \Oy, \lambda, \pi\, \vert\, \dY, \dX) & \propto & p(\dY\, \vert\, \dX, \D, \Oy)\, p(\D\, \vert\, \Oy, \lambda, \pi)\, p(\lambda)\, p(\Oy)\, p(\pi) \nonumber \\
 & \propto & \vert \Oy \vert^{\frac{n}{2}}\, \exp\!\left(\! -\frac{1}{2}\, \left\Vert (\dY + \dX\, \D^t\, \Oyi)\, \Oy^{\frac{1}{2}} \right\Vert_{F}^2 \right) \nonumber \\
 & & \hsp \times~ \prod_{i=1}^{p} \Bigg[ \frac{1-\pi}{\sqrt{\lambda_i^{\, q}\, \vert \Oy \vert}}\, \exp\!\left(\! -\frac{\D_i^t\, \Oyi \D_i}{2\, \lambda_i} \right)\! \ind_{\{ \D_i\, \neq\, 0 \}} \nonumber \\
 & & \hsp \hsp \hsp \hsp \hsp +~ \pi\, \ind_{\{ \D_i\, =\, 0 \}} \Bigg] \lambda_i^{\frac{1}{2}(q+1)-1}\, \de^{-\ell_i\, \lambda_i} \nonumber \\
 & & \hsp \times~ \vert \Oy \vert^{\frac{u-q-1}{2}}\, \exp\!\left(\! -\frac{\tr(V^{-1}\, \Oy)}{2} \right) \pi^{\, a-1}\, (1-\pi)^{\, b-1}.
\end{eqnarray}
On the one hand, exploiting the cyclic property of the trace, a tedious calculation shows that, for all $1 \leq i \leq p$,
\begin{eqnarray}
\label{DecompS}
\left\Vert (\dY + \dX\, \D^t\, \Oyi)\, \Oy^{\frac{1}{2}} \right\Vert_{F}^2 & = & \tr( \dY^{\, t}\, \dY\, \Oy) + 2\, \tr(\dX^t\, \dY \D) + \tr( \dX^t\, \dX\, \D^t\, \Oyi \D) \nonumber \\
 & = & \Vert \dX_i \Vert^{\, 2}\, \D_i^t\, \Oyi \D_i + 2 \sum_{j\, \neq\, i} \langle \dX_i,\, \dX_j \rangle\, \D_j^t\, \Oyi \D_i + 2\, \dX_i^{\, t}\, \dY \D_i + T_{\neq\, i}
\end{eqnarray}
where the term $T_{\neq\, i}$ does not depend on $\D_i$. Thus,
\begin{eqnarray}
\label{CondDiS}
p(\D_i\, \vert\, \Theta_{\D_i}) & \propto & \exp\!\left(\! -\frac{1}{2}\, \Vert \dX_i \Vert^{\, 2}\, \D_i^t\, \Oyi \D_i - \sum_{j\, \neq\, i} \langle \dX_i,\, \dX_j \rangle\,  \D_j^t\, \Oyi \D_i - \dX_i^{\, t}\, \dY \D_i \right) \nonumber \\
 & & \hsp \times~ \left[ \frac{1-\pi}{\sqrt{\lambda_i^{\, q}\, \vert \Oy \vert}}\,\, \exp\!\left(\! -\frac{\D_i^t\, \Oyi \D_i}{2\, \lambda_i} \right)\! \ind_{\{ \D_i\, \neq\, 0 \}} + \pi\, \ind_{\{ \D_i\, =\, 0 \}} \right] \nonumber \\
 & = & \exp\!\left(\! -\frac{1}{2}\, (\D_i + s_i\, H_i)^{\, t}\, (s_i\, \Oy)^{-1}\, (\D_i + s_i\, H_i) \right) \nonumber \\
 & & \hsp \times~ \exp\!\left(\! \frac{s_i\, H_i^{\, t}\, \Oyi\, H_i}{2} \right) \frac{1-\pi}{\sqrt{\lambda_i^{\, q}\, \vert \Oy \vert}}\, \ind_{\{ \D_i\, \neq\, 0 \}} + \pi\, \ind_{\{ \D_i\, =\, 0 \}}
\end{eqnarray}
for all $1 \leq i \leq p$, where
\begin{equation*}
H_i = \Oy\, \dY^{\, t}\, \dX_i + \sum_{j\, \neq\, i} \langle \dX_i,\, \dX_j \rangle\, \D_j \hsp \text{and} \hsp s_i = \frac{\lambda_i}{1 + \lambda_i\, \Vert \dX_i \Vert^{\, 2}}.
\end{equation*}
This is still a multivariate Gaussian spike-and-slab distribution such that, by renormalizing, the spike has probability
\begin{equation*}
p_i = \dP(\D_i = 0\, \vert\, \Theta_{\D_i}) = \frac{\pi}{\pi + (1-\pi)\, (1 + \lambda_i\, \Vert \dX_i \Vert^{\, 2})^{-\frac{q}{2}}\, \exp\!\Big( \frac{s_i\, H_i^{\, t}\, \Oyi H_i}{2} \Big)}.
\end{equation*}
On the other hand, coming back to \eqref{DecompS}, we can also write
\begin{equation*}
\left\Vert (\dY + \dX\, \D^t\, \Oyi)\, \Oy^{\frac{1}{2}} \right\Vert_{F}^2 = \tr( \dY^{\, t}\, \dY\, \Oy) + \tr( \D\, \dX^t\, \dX\, \D^t\, \Oyi) + T_{\neq\, y}
\end{equation*}
where $T_{\neq\, y}$ does not depend on $\Oy$. That leads, \textit{via} \eqref{FullS}, to
\begin{eqnarray}
\label{CondOyS}
p(\Oy\, \vert\, \Theta_{\Oy}) & \propto & \vert \Oy \vert^{\frac{n-p+N_0 +u-q-1}{2}}\, \exp\!\Bigg(\!-\frac{1}{2}\, \tr( (\dY^{\, t}\, \dY + V^{-1})\, \Oy) \nonumber \\
 & & \hsp -~ \frac{1}{2}\! \left( \tr(\D\, \dX^t\, \dX\, \D^t\, \Oyi) + \sum_{\D_i\, \neq\, 0} \frac{\D_i^t\, \Oyi \D_i}{\lambda_i} \right)\! \Bigg) \nonumber \\
 & = & \vert \Oy \vert^{\frac{n-p+N_0 +u-q-1}{2}}\, \exp\!\left(\! -\frac{1}{2}\, \tr\big( (\dY^{\, t}\, \dY + V^{-1})\, \Oy + \D\, (\dX^t\, \dX + D_\lambda^{-1})\, \D^t\, \Oyi \big) \right)
\end{eqnarray}
where $N_0$ is given in \eqref{N0} and $D_\lambda = \diag(\lambda_1, \hdots, \lambda_p)$. Finally, it is easy to see that, for all $1 \leq i \leq p$,
\begin{equation}
\label{CondLiS}
p(\lambda_i\, \vert\, \Theta_{\lambda_i}) ~\propto~ \frac{1}{\sqrt{\lambda_i}}\, \exp\!\left(\! -\frac{\D_i^t\, \Oyi \D_i}{2\, \lambda_i} -\ell_i\, \lambda_i \right)\! \ind_{\{ \D_i\, \neq\, 0 \}} + \lambda_i^{\frac{1}{2}(q+1)-1}\, \de^{-\ell_i\, \lambda_i} \, \ind_{\{ \D_i\, =\, 0 \}}
\end{equation}
whereas
\begin{equation}
\label{CondPiS}
p(\pi\, \vert\, \Theta_{\pi}) ~\propto~ \pi^{N_0+a-1}\, (1-\pi)^{p-N_0+b-1}.
\end{equation}
We recognize in \eqref{CondDiS}, \eqref{CondOyS}, \eqref{CondLiS} and \eqref{CondPiS} the announced conditional posterior distributions, which concludes the proof. \qed

\subsection{The group-sparse setting: proof of Proposition \ref{PropPostGS}}
\label{SecPostGS}

The full posterior distribution of the parameters conditional on $\dX$ and $\dY$ satisfies
\begin{eqnarray}
\label{FullGS}
p(\D, \Oy, \lambda, \pi\, \vert\, \dY, \dX) & \propto & p(\dY\, \vert\, \dX, \D, \Oy)\, p(\D\, \vert\, \Oy, \lambda, \pi)\, p(\lambda)\, p(\Oy)\, p(\pi) \nonumber \\
 & \propto & \vert \Oy \vert^{\frac{n}{2}}\, \exp\!\left(\! -\frac{1}{2}\, \left\Vert (\dY + \dX\, \D^t\, \Oyi)\, \Oy^{\frac{1}{2}} \right\Vert_{F}^2 \right) \nonumber \\
 & & \hsp \times~ \prod_{g=1}^m \Bigg[ \frac{1-\pi}{\sqrt{\lambda_g^{\, q\, \kappa_g}\, \vert \Oy \vert^{\kappa_g}}}\, \exp\!\left(\! -\frac{\tr(\gD_g^t\, \Oyi \gD_g)}{2\, \lambda_g} \right)\! \ind_{\{ \D_g\, \neq\, 0 \}} \nonumber \\
 & & \hsp \hsp \hsp \hsp \hsp +~ \pi\, \ind_{\{ \D_g\, =\, 0 \}} \Bigg] \lambda_g^{\frac{1}{2}(q\, \kappa_g+1) - 1}\, \de^{-\ell_g\, \lambda_g} \nonumber \\
 & & \hsp \times~ \vert \Oy \vert^{\frac{u-q-1}{2}}\, \exp\!\left(\! -\frac{\tr(V^{-1}\, \Oy)}{2} \right) \pi^{\, a-1}\, (1-\pi)^{\, b-1}.
\end{eqnarray}
Like in the previous proof, a first important step is to note that, for all $1 \leq g \leq m$,
\begin{eqnarray}
\label{DecompGS}
\left\Vert (\dY + \dX\, \D^t\, \Oyi)\, \Oy^{\frac{1}{2}} \right\Vert_{F}^2 & = & \left\Vert \dY\, \Oy^{\frac{1}{2}} + \sum_{j=1}^m \gdX_j\, \gD_j^t\, \Oy^{-\frac{1}{2}} \right\Vert_{F}^2 \nonumber \\
 & = & \Vert \gdX_g\, \gD_g^t\, \Oy^{-\frac{1}{2}} \Vert_{F}^2 + 2\, \sum_{j\, \neq\, g} \tr(\gD_j\, \gdX_j^{\, t}\, \gdX_g\, \gD_g^t\, \Oyi) \nonumber \\
 & & \hsp \hsp \hsp \hsp \hsp +~ 2\, \tr(\gdX_g^{\, t}\, \dY\, \gD_g) + T_{\neq\, g}
\end{eqnarray}
where the term $T_{\neq\, g}$ does not depend on $\gD_g$. Thus, after a tedious calculation exploiting the cyclic property of the trace, one can obtain the factorization
\begin{eqnarray}
\label{CondDgGS}
p(\gD_g\, \vert\, \Theta_{\D_g}) & \propto & \exp\!\left(\! -\frac{1}{2}\, \Vert \gdX_g\, \gD_g^t\, \Oy^{-\frac{1}{2}} \Vert_{F}^2 - \sum_{j\, \neq\, g} \tr(\gD_j\, \gdX_j^{\, t}\, \gdX_g\, \gD_g^t\, \Oyi) - \tr(\gdX_g^{\, t}\, \dY\, \gD_g) \right) \nonumber \\
 & & \hsp \times~ \left[ \frac{1-\pi}{\sqrt{\lambda_g^{\, q\, \kappa_g}\, \vert \Oy \vert^{\kappa_g}}}\,\, \exp\!\left(\! -\frac{\tr(\gD_g^t\, \Oyi \gD_g)}{2\, \lambda_g} \right)\! \ind_{\{ \D_g\, \neq\, 0 \}} + \pi\, \ind_{\{ \D_g\, =\, 0 \}} \right] \nonumber \\
 & = & \exp\!\left(\! -\frac{1}{2}\, \tr\big( S_g^{-1} (\gD_g + H_g\, S_g)^t\, \Oyi (\gD_g + H_g\, S_g) \big) \right) \nonumber \\
 & & \hsp \times~ \exp\!\left( \frac{\tr(H_g^{\, t}\, \Oyi\, H_g\, S_g)}{2} \right) \frac{1-\pi}{\sqrt{\lambda_g^{\, q\, \kappa_g}\, \vert \Oy \vert^{\kappa_g}}}\, \ind_{\{ \D_g\, \neq\, 0 \}} + \pi\, \ind_{\{ \D_g\, =\, 0 \}}
\end{eqnarray}
for all $1 \leq g \leq m$, where
\begin{equation*}
H_g = \Oy\, \dY^{\, t}\, \gdX_g + \sum_{j\, \neq\, g} \gD_j\, \gdX_j^{\, t}\, \gdX_g \hsp \text{and} \hsp S_g = \lambda_g \big( I_{\kappa_g} + \lambda_g\, \gdX_g^{\, t}\, \gdX_g \big)^{-1}.
\end{equation*}
We recognize the announced Gaussian spike-and-slab distribution, and the probability of the spike is given, after renormalization, by
\begin{equation*}
p_g = \dP(\gD_g = 0\, \vert\, \Theta_{\D_g}) = \frac{\pi}{\pi + (1-\pi)\, \vert I_{\kappa_g} + \lambda_g\, \gdX_g^{\, t}\, \gdX_g \vert^{-\frac{q}{2}}\, \exp\!\Big( \frac{\tr(H_g^{\, t}\, \Oyi H_g\, S_g)}{2} \Big)}.
\end{equation*}
Following the same lines as the ones used to establish \eqref{CondOyS}, we obtain from \eqref{FullGS} the conditional distribution
\begin{eqnarray}
\label{CondOyGS}
p(\Oy\, \vert\, \Theta_{\Oy}) & \propto & \vert \Oy \vert^{\frac{n-p+N_0 +u-q-1}{2}}\, \exp\!\Bigg(\!-\frac{1}{2}\, \tr( (\dY^{\, t}\, \dY + V^{-1})\, \Oy) \nonumber \\
 & & \hsp -~ \frac{1}{2} \Bigg( \tr(\D\, \dX^t\, \dX\, \D^t\, \Oyi) + \sum_{\D_g\, \neq\, 0} \frac{\tr(\gD_g^t\, \Oyi \gD_g)}{\lambda_g} \Bigg)\! \Bigg) \nonumber \\
 & = & \vert \Oy \vert^{\frac{n-p+N_0 +u-q-1}{2}}\, \exp\!\left(\! -\frac{1}{2}\, \tr\big( (\dY^{\, t}\, \dY + V^{-1})\, \Oy + \D\, (\dX^t\, \dX + D_\lambda^{-1})\, \D^t\, \Oyi \big) \right)
\end{eqnarray}
where $D_\lambda = \diag(\lambda_1, \hdots, \lambda_1, \hdots, \lambda_m, \hdots, \lambda_m)$ with each $\lambda_g$ duplicated $\kappa_g$ times, and since we can note that, due to the continuous nature of $\D\, \vert \{ \D \neq 0 \}$,
\begin{equation*}
\sum_{g=1}^{m} \kappa_g \ind_{\{ \D_g\, \neq\, 0 \}} = p-N_0
\end{equation*}
for $N_0$ given in \eqref{N0}. Next, we obtain in a simpler way that, for all $1 \leq g \leq m$,
\begin{eqnarray}
\label{CondLgGS}
p(\lambda_g\, \vert\, \Theta_{\lambda_g}) & \propto&  \frac{1}{\sqrt{\lambda_g}}\, \exp\!\left(\! -\frac{\tr(\gD_g^t\, \Oyi \gD_g)}{2\, \lambda_g} -\ell_g\, \lambda_g \right)\! \ind_{\{ \D_g\, \neq\, 0 \}} \nonumber \\
& & \hsp +~ \lambda_g^{\frac{1}{2}(q\, \kappa_g+1) - 1}\, \de^{-\ell_g\, \lambda_g} \, \ind_{\{ \D_g\, =\, 0 \}}.
\end{eqnarray}
Finally,
\begin{equation}
\label{CondPiGS}
p(\pi\, \vert\, \Theta_{\pi}) ~\propto~ \pi^{G_0+a-1}\, (1-\pi)^{m-G_0+b-1}
\end{equation}
where $G_0$ is defined in \eqref{G0}. We can check that the conditional distributions \eqref{CondDgGS}, \eqref{CondOyGS}, \eqref{CondLgGS} and \eqref{CondPiGS} correspond to the ones announced in the proposition, which concludes the proof. \qed

\subsection{The sparse-group-sparse setting: proof of Proposition \ref{PropPostSGS}}
\label{SecPostSGS}

The full posterior distribution of the parameters conditional on $\dX$ and $\dY$ satisfies
\begin{eqnarray}
\label{FullSGS}
p(\D, \Oy, \nu, \lambda, \pi\, \vert\, \dY, \dX) & \propto & p(\dY\, \vert\, \dX, \D, \Oy)\, p(\D\, \vert\, \Oy, \nu, \lambda, \pi)\, p(\nu)\, p(\lambda)\, p(\Oy)\, p(\pi) \nonumber \\
 & \propto & \vert \Oy \vert^{\frac{n}{2}}\, \exp\!\left(\! -\frac{1}{2}\, \left\Vert (\dY + \dX\, \D^t\, \Oyi)\, \Oy^{\frac{1}{2}} \right\Vert_{F}^2 \right) \nonumber \\
 & & \hsp \times~ \prod_{g=1}^m \Bigg[ \big( (1-\pi_1)\, P_g\, \ind_{\{ \D_g\, \neq\, 0 \}} + \pi_1\, \ind_{\{ \D_g\, =\, 0 \}} \big) \nonumber \\
 & & \hsp \hsp \hsp \hsp \hsp \times~ \lambda_g^{\frac{1}{2} (q\, \kappa_g+1)-1}\, \de^{-\gamma_g\, \lambda_g} \prod_{i=1}^{\kappa_g} \nu_{gi}^{\frac{1}{2} (q+1)-1}\, \de^{-\ell_{gi}\, \nu_{gi}} \Bigg] \nonumber \\
 & & \hsp \times~ \vert \Oy \vert^{\frac{u-q-1}{2}}\, \exp\!\left(\! -\frac{\tr(V^{-1}\, \Oy)}{2} \right) \prod_{j=1}^2 \pi_j^{\, a_j-1}\, (1-\pi_j)^{\, b_j-1}
\end{eqnarray}
where, for $1 \leq g \leq m$,
\begin{equation*}
P_g = \prod_{i=1}^{\kappa_g} \left[ \frac{1-\pi_2}{\sqrt{(\nu_{gi}\, \lambda_g)^{\, q}\, \vert \Oy \vert}}\, \exp\!\left(\! -\frac{\D_{gi}^t\, \Oyi \D_{gi}}{2\, \nu_{gi}\, \lambda_g} \right)\! \ind_{\{ \D_{gi}\, \neq\, 0 \}} + \pi_2\, \ind_{\{ \D_{gi}\, =\, 0 \}} \right].
\end{equation*}
Using the same decompositions as \eqref{DecompS} or \eqref{DecompGS}, the full posterior distribution given above leads to
\begin{eqnarray}
\label{CondDgiSGS}
p(\D_{gi}\, \vert\, \Theta_{\D_{gi}}) & \propto & \exp\!\left(\! -\frac{1}{2}\, \Vert \dX_{gi} \Vert^{\, 2}\, \D_{gi}^t\, \Oyi \D_{gi} - \sum_{h, j\, \neq\, g, i} \langle \dX_{gi},\, \dX_{hj} \rangle\,  \D_{hj}^t\, \Oyi \D_{gi} - \dX_{gi}^{\, t}\, \dY \D_{gi} \right) \nonumber \\
 & & \hsp \times~ \Bigg[ (1-\pi_1) \Bigg[ \frac{1-\pi_2}{\sqrt{(\nu_{gi}\, \lambda_g)^{\, q}\, \vert \Oy \vert}}\,\, \exp\!\left(\! -\frac{\D_{gi}^t\, \Oyi \D_{gi}}{2\, \nu_{gi}\, \lambda_g} \right)\! \ind_{\{ \D_{gi}\, \neq\, 0 \}} \nonumber \\
 & & \hsp \hsp \hsp \hsp \hsp +~ \pi_2\, \ind_{\{ \D_{gi}\, =\, 0 \}} \Bigg] \ind_{\{ \D_g\, \neq\, 0 \}} + \pi_1\, \ind_{\{ \D_g\, =\, 0 \}} \Bigg] \nonumber \\
 & = & \exp\!\left(\! -\frac{1}{2}\, (\D_{gi} + s_{gi}\, H_{gi})^{\, t}\, (s_{gi}\, \Oy)^{-1}\, (\D_{gi} + s_{gi}\, H_{gi}) \right) \nonumber \\
 & & \hsp \times~ \exp\!\left(\! \frac{s_{gi}\, H_{gi}^{\, t}\, \Oyi\, H_{gi}}{2} \right) \frac{(1-\pi_1)\, (1-\pi_2)}{\sqrt{(\nu_{gi}\, \lambda_g)^{\, q}\, \vert \Oy \vert}}\, \ind_{\{ \D_{gi}\, \neq\, 0 \}} \nonumber \\
 & & \hsp +~ \big( (1-\pi_1)\, \pi_2\, \ind_{\{ \D_{g \backslash i}\, \neq\, 0 \}}\, + \pi_1\, \ind_{\{ \D_{g \backslash i}\, =\, 0 \}} \big)\, \ind_{\{ \D_{gi}\, =\, 0 \}}
\end{eqnarray}
for $1 \leq g \leq m$ and $1 \leq i \leq \kappa_g$, where $\gD_{g \backslash i}$ is $\gD_g$ deprived of $\D_{gi}$,
\begin{equation*}
H_{gi} = \Oy\, \dY^{\, t}\, \dX_{gi} + \sum_{h, j\, \neq\, g, i} \langle \dX_{gi},\, \dX_{hj} \rangle\, \D_{hj} \hsp \text{and} \hsp s_{gi} = \frac{\nu_{gi}\, \lambda_g}{1 + \nu_{gi}\, \lambda_g\, \Vert \dX_{gi} \Vert^{\, 2}}.
\end{equation*}
Here, we used the binary equalities stemming from $\{ \D_{gi}\, \neq\, 0 \} \cap \{\gD_g\, \neq\, 0 \} = \{ \D_{gi}\, \neq\, 0 \}$, $\{ \D_{gi}\, =\, 0 \} \cap \{ \gD_g\, \neq\, 0 \} = \{ \D_{gi}\, =\, 0 \} \cap \{ \gD_{g \backslash i}\, \neq\, 0 \}$ and $\{ \D_{gi}\, =\, 0 \} \cap \{ \gD_g\, =\, 0 \} = \{ \D_{gi}\, =\, 0 \} \cap \{ \gD_{g \backslash i}\, =\, 0 \}$, which turn out to be very useful to separate $\D_{gi}$ and $\Theta_{\D_{gi}}$. This is characteristic of a multivariate Gaussian spike-and-slab distribution. By renormalizing, one can see that the spike has probability
\begin{equation*}
p_{gi} = \dP(\D_{gi} = 0\, \vert\, \Theta_{\D_{gi}}) = \frac{\rho_{gi}}{\rho_{gi} + (1-\pi_1)\, (1-\pi_2)\, (1 + \nu_{gi}\, \lambda_g\, \Vert \dX_{gi} \Vert^{\, 2})^{-\frac{q}{2}}\, \exp\!\Big( \frac{s_{gi}\, H_{gi}^{\, t}\, \Oyi\, H_{gi}}{2} \Big)}
\end{equation*}
with
\begin{equation*}
\rho_{gi} = (1-\pi_1)\, \pi_2\, \ind_{\{ \D_{g \backslash i}\, \neq\, 0 \}}\, + \pi_1\, \ind_{\{ \D_{g \backslash i}\, =\, 0 \}}.
\end{equation*}
Next, following \eqref{FullSGS} and the reasoning used to establish \eqref{CondOyS}, we may also write
\begin{eqnarray}
\label{CondOySGS}
p(\Oy\, \vert\, \Theta_{\Oy}) & \propto & \vert \Oy \vert^{\frac{n-p+N_0 +u-q-1}{2}}\, \exp\!\Bigg(\!-\frac{1}{2}\, \tr( (\dY^{\, t}\, \dY + V^{-1})\, \Oy) \nonumber \\
 & & \hsp -~ \frac{1}{2} \Bigg( \tr(\D\, \dX^t\, \dX\, \D^t\, \Oyi) + \sum_{\D_{gi}\, \neq\, 0} \frac{\D_{gi}^t\, \Oyi \D_{gi}}{\nu_{gi}\, \lambda_g} \Bigg)\! \Bigg) \nonumber \\
 & = & \vert \Oy \vert^{\frac{n-p+N_0 +u-q-1}{2}}\, \exp\!\left(\! -\frac{1}{2}\, \tr\big( (\dY^{\, t}\, \dY + V^{-1})\, \Oy + \D\, (\dX^t\, \dX + D_{\lambda \nu}^{-1})\, \D^t\, \Oyi \big) \right)
\end{eqnarray}
where $N_0$ is given in \eqref{N0} and $D_{\lambda \nu} = \diag(\nu_{11} \lambda_1, \hdots, \nu_{1\kappa_1} \lambda_1, \hdots, \nu_{m1} \lambda_m, \hdots, \nu_{m\kappa_m} \lambda_m)$. The shrinkage parameters $\nu$ and $\lambda$ are easier to handle. For $1 \leq g \leq m$ and $1 \leq i \leq \kappa_g$,
\begin{eqnarray}
\label{CondNgiSGS}
p(\nu_{gi}\, \vert\, \Theta_{\nu_{gi}}) & \propto & \frac{1}{\sqrt{\nu_{gi}}}\, \exp\!\left(\! -\frac{\D_{gi}^t\, \Oyi \D_{gi}}{2\, \nu_{gi}\, \lambda_g} -\ell_{gi}\, \nu_{gi} \right)\! \ind_{\{ \D_{gi}\, \neq\, 0 \}} \nonumber \\
& & \hsp +~ \nu_{gi}^{\frac{1}{2}(q+1)-1}\, \de^{-\ell_{gi}\, \nu_{gi}} \, \ind_{\{ \D_{gi}\, =\, 0 \}}
\end{eqnarray}
whereas
\begin{eqnarray}
\label{CondLgSGS}
p(\lambda_g\, \vert\, \Theta_{\lambda_g}) & \propto & \lambda_g^{\frac{q N_{0g} - 1}{2}} \exp\!\left(\! -\frac{\tr(D_{\nu_g}^{-1} \gD_g^t\, \Oyi \gD_g)}{2\, \lambda_g} -\gamma_g\, \lambda_g \right)\! \ind_{\{ \D_g\, \neq\, 0 \}} \nonumber \\
 & & \hsp +~ \lambda_g^{\frac{1}{2}(q\, \kappa_g+1)-1}\, \de^{-\gamma_g\, \lambda_g} \, \ind_{\{ \D_g\, =\, 0 \}}
\end{eqnarray}
where $N_{0g}$ is defined in \eqref{N0g} and $D_{\nu_g} = \diag(\nu_{g1}, \hdots, \nu_{g\kappa_g})$. Finally,
\begin{equation}
\label{CondPi1SGS}
p(\pi_1\, \vert\, \Theta_{\pi_1}) ~\propto~ \pi_1^{G_0+a_1-1}\, (1-\pi_1)^{m-G_0+b_1-1}
\end{equation}
and
\begin{equation}
\label{CondPi2SGS}
p(\pi_2\, \vert\, \Theta_{\pi_2}) ~\propto~ \pi_2^{J_0+a_2-1}\, (1-\pi_2)^{p-N_0+b_2-1}
\end{equation}
where $G_0$ and $J_0$ are given in \eqref{G0} and \eqref{J0}, respectively. For the latter result, we used the fact that the number of non-zero columns in the non-zero groups must coincide with the number of non-zero columns of $\D$, that is $p-N_0$. Like in the previous proofs, we recognize the announced conditional distributions in \eqref{CondDgiSGS}, \eqref{CondOySGS}, \eqref{CondNgiSGS}, \eqref{CondLgSGS}, \eqref{CondPi1SGS} and \eqref{CondPi2SGS}. That concludes these tedious calculations. \qed

\subsection{Proof of Proposition \ref{PropSupRecGS}}
\label{SecSupRecGS}

The result is obtained by following the steps of the proof of Thm 2.1 in \cite{YangNarisetty20} but, beforehand, we need to clarify a few points to extend the reasoning of the authors from $q = 1$ to $q \geq 1$ and take into account the adaptative shrinkage. For any model $k$, let $\cK = \{ k \text{ is selected} \}$ so that $\cK = \cT$ when the true model $t$ is considered. First, recall that $\lambda$ and $\pi$ are fixed and rewrite \eqref{FullGS} like
\begin{eqnarray}
\label{FullGSSup}
\dP_{\Delta}(\cK\, \vert\, \dY, \dX, \Oy) & \propto & \exp\!\left(\! -\frac{1}{2}\, \left\Vert (\dY + \dX_{(k)}\, \D_{(k)}^t\, \Oyi)\, \Oy^{\frac{1}{2}} \right\Vert_{F}^2 \right) \nonumber \\
 & & \hsp \times~ \frac{(1-\pi)^{\vert k \vert}}{\pi^{\vert k \vert}\, \sqrt{\vert \Lambda_k \vert^q\, \vert \Oy \vert^{k_r}}}\, \exp\!\left(\! -\frac{\tr(\D_{(k)}^t\, \Oyi \D_{(k)}\, D_k^{-1})}{2} \right) \nonumber \\
 & \propto & \frac{(1-\pi)^{\vert k \vert}}{\pi^{\vert k \vert}\, \sqrt{\vert \Lambda_k \vert^q\, \vert \Oy \vert^{k_r}}}\, \exp\!\left(\! \frac{\tr( \Dt_{(k)}\, F_k\, \Dt_{(k)}^t\, \Oyi )}{2} \right) \nonumber \\
  & & \hsp \times~ \exp\!\left(\! -\frac{1}{2}\, \tr\!\left( (\D_{(k)} - \Dt_{(k)})\, F_k\, (\D_{(k)} - \Dt_{(k)})^t\, \Oyi \right) \right)
\end{eqnarray}
where $F_k = D_k^{-1} + \dX_{(k)}^t\, \dX_{(k)}$, $D_k = \diag((\lambda_\ell, \hdots, \lambda_\ell)_{\ell\, \in\, k})$ with each $\lambda_\ell$ duplicated $\kappa_\ell$ times, $k_r = \Vert (\kappa_\ell)_{\ell\, \in\, k} \Vert_1$, $\Lambda_k = \diag((\lambda_\ell^{\kappa_\ell})_{\ell\, \in\, k})$ and
\begin{equation*}
\Dt_{(k)} = -\Oy\, \dY^t\, \dX_{(k)}\, F_k^{\, -1}.
\end{equation*}
Then, integrating over $\D_{(k)}$, it follows (see Def. \ref{DefN} with $\Sigma_1 = \Oy$ and $\Sigma_2 = F_k^{\, -1}$) that
\begin{eqnarray*}
\dP(\cK\, \vert\, \dY, \dX, \Oy) & = & \int_{\dR^{q \times \kappa_k}} \dP_{\Delta}(\cK\, \vert\, \dY, \dX, \Oy, \lambda, \pi)\, \dd \D_{(k)} \\
 & \propto & \frac{(1-\pi)^{\vert k \vert}}{\pi^{\vert k \vert}\, \sqrt{\vert \Lambda_k \vert^q\, \vert F_k \vert^q}}\, \exp\!\left(\! \frac{\tr( \Dt_{(k)}\, F_k\, \Dt_{(k)}^t\, \Oyi )}{2} \right) \\
 & \propto & \left( \frac{1-\pi}{\pi} \right)^{\! \vert k \vert} \vert \Lambda_k \vert^{-\frac{q}{2}}\, \vert F_k \vert^{-\frac{q}{2}}\, \exp\!\left( -\frac{1}{2}\, \tr\big( \dY^{*\, t}\, (I_n - \dX_{(k)}\, F_k^{\, -1}\, \dX_{(k)}^t)\, \dY^* \big) \right) \\
 & = & \left( \frac{1-\pi}{\pi} \right)^{\! \vert k \vert} \vert \Lambda_k \vert^{-\frac{q}{2}}\, \vert F_k \vert^{-\frac{q}{2}}\, \exp\!\left( -\frac{1}{2}\, \left( \textnormal{RSS}_k(\Dt^*_{(k)}) + \big\Vert \Dt^*_{(k)}\, D_k^{-\frac{1}{2}} \big\Vert^2_F \right) \right)
\end{eqnarray*}
where $\dY^* = \dY\, \Oy^{\frac{1}{2}}$, $\Dt^*_{(k)} = \Oy^{-\frac{1}{2}}\, \Dt_{(k)}$ and $\textnormal{RSS}_k\, : H \in \dR^{q \times \kappa_k} \mapsto \Vert \dY^* - \dX_{(k)}\, H^{\, t} \Vert_F^2$ is the residual sum of squares function in the renormalized linear model indexed by $k$, that is
\begin{equation*}
\dY^* = -\dX_{(k)}\, \D_{(k)}^t\, \Oy^{-\frac{1}{2}} + E^*
\end{equation*}
with $E^* = E\, \Oy^{\frac{1}{2}} \sim \cMN_{n \times q}(0,\, I_{n},\, I_{q})$. Thus, the so-called posterior ratio between any false model $k$ and $t$ is given by
\begin{equation*}
\textnormal{PR}(k, t) = \frac{\dP(\cK\, \vert\, \dY, \dX, \Oy)}{\dP(\cT\, \vert\, \dY, \dX, \Oy)} = \frac{Q_k}{Q_t} \left( \frac{1-\pi}{\pi} \right)^{\! \vert k \vert - \vert t \vert} \de^{-\frac{1}{2}\, (\wt{R}_k - \wt{R}_t)}
\end{equation*}
with $Q_k = \vert \Lambda_k \vert^{-\frac{q}{2}}\, \vert F_k \vert^{-\frac{q}{2}}\,$ and $\wt{R}_k = \textnormal{RSS}_k(\Dt^*_{(k)}) + \Vert \Dt^*_{(k)}\, D_k^{-\frac{1}{2}} \Vert^2_F$, using the notation of \cite{YangNarisetty20}. In particular, due to the generalized ridge penalty,
\begin{equation}
\Dt^*_{(k)} = \arg \min_{H} ~ \left( \textnormal{RSS}_k(H) + \big\Vert H D_k^{-\frac{1}{2}} \big\Vert^2_F \right)
\end{equation}
so that for nested models $k_1$ and $k_2$ (with $k_1 \subseteq k_2$), we must have $\wt{R}_{k_2} \leq \wt{R}_{k_1}$. Let also $R_k = \Vert (I_n - \Pi_{(k)})\, \dY^* \Vert_F^2 = \Vert (I_q \otimes (I_n - \Pi_{(k)}))\, \vec(\dY^*) \Vert_2^2$. Cochran's theorem entails the chi-squared distributions $R_t \sim \chi^2(q\, (n-r_t))$ and $R_t - R_k \sim \chi^2(q\, (r_k-r_t))$ for any `bigger' model $k \supset t$ and $q \geq 1$. Combining all these preliminary considerations, the strategy of \cite{YangNarisetty20} now applies and leads, under our revised hypotheses, to
\begin{equation*}
\frac{1-\dP(\cT\, \vert\, \dY, \dX, \Oy)}{\dP(\cT\, \vert\, \dY, \dX, \Oy)} = \sum_{k\, \neq\, t} \textnormal{PR}(k, t) \limp 0.
\end{equation*} \qed

\section{Empirical results}
\label{SecEmp}

In this section, let us call (s), (gs) and (sgs) the related settings, and let us denote by (ad) the adaptative shrinkage and by (gl) the global shrinkage. First of all, these models contain many hyperparameters that have to be carefully tuned. Our experiments showed that, unsurprinsingly, the results are strongly impacted by the prior amount of shrinkage on $\D$, driven by $\ell$ and even by $\gamma$ for (sgs). Apart from the usual cross-validation procedures, we could stay in line with our Bayesian approach and suggest conjugate Gamma hyperpriors. This is very easy to implement, but the hyperparameters are now replaced by other hyperparameters and the same questions arise. Instead, like in \cite{XuGosh15} and \cite{LiquetEtAl17}, we follow the idea of \cite{ParkCasella08} and we use a Monte-Carlo EM algorithm. By way of example, from the full posterior probability \eqref{FullS} and since $\lambda_i \sim \Gamma(q,\, \ell_i)$ for all $i$, it is not hard to see that, with (s),
\begin{equation*}
\ln p(\D, \Oy, \lambda, \pi\, \vert\, \dY, \dX) = \sum_{i=1}^p (q\, \ln \ell_i - \ell_i\, \lambda_i) + T_{\neq\, \ell}
\end{equation*}
where the term $T_{\neq\, \ell}$ does not depend on $\ell$. Thus, the $k$-th iteration of the EM algorithm should lead to
\begin{equation*}
\ell^{\, (k)}_i = \frac{\frac{1}{2} (q+1)}{\dE^{(k-1)}[\lambda_i\, \vert\, \dY,\, \dX]} \hsp \text{and} \hsp \ell^{\, (k)} = \frac{\frac{p}{2} (q+1)}{\sum_{i=1}^p \dE^{(k-1)}[\lambda_i\, \vert\, \dY,\, \dX]}
\end{equation*}
for the adaptative shrinkage and the global shrinkage ($\lambda_i = \lambda$), respectively. The intractable conditional expectations are then estimated with the help of the Gibbs samples. For (gs), the results are mainly the same as above (replace $q+1$ by $q \kappa_g+1$ in the first case, $p(q+1)$ by $q p+m$ in the second case and consider $1 \leq g \leq m$ instead of $1 \leq i \leq p$), and similar results also follow with (sgs). Recall that our definitions of the adaptative and global shrinkages are given in the corresponding sections, in the description of the hierarchical models. The tuning of $u$ and $V$ (or $v$) is actually trickier. Because $\dE[\Oy] = u V$, we set $V = \frac{1}{u} I_q$ and $u$ is conveniently chosen to be the smallest integer such that $\Oy$ is (almost surely) invertible, that is $u = q$ (see \textit{e.g.} \cite{BrownEtAl98}). This is particularly adapted when the dataset is standardized. Finally, $a$ and $b$ reflect the degree of sparsity to introduce in the direct links. We can set $a \gg b$ to promote sparse settings, which is potentially interesting when $p \gg n$, but $a=b=1$ is a standard non-informative choice and $a < b$ may also be useful for variable selection (see \textit{e.g.} the real dataset of Section \ref{SecEmpReal}). The posterior median is used to estimate $\D$ and get sparsity whereas the posterior mean is used to estimate $\Oy$. Due to the huge amount of calculations in the simulations, the estimations are made on the basis of $3000$ iterations of the sampler in which the first $2000$ are burn-ins. This is revised upwards for the real data ($10000$ iterations with $5000$ burn-ins).

\begin{rem}
To the best of our knowledge, there is no simple way to sample from the $\cMGIG_d$ distribution as soon as $d > 1$. The recent method described in Sec. 3.3.2 of \cite{FangEtAl20}, relying on the Matsumoto-Yor property (see Thm. 3.1 of \cite{MassamWesolowski06}) to get a $\cMGIG_d$ sample from the very standard $\cGIG$ and $\cW_d$ distributions, is unfortunately inapplicable in our contexte. Indeed, for example in the sparse setting, that would require finding $z \in \dR^q$ such that $\dY^{\, t}\, \dY + V^{-1} = b\, z z^t$ for some $b > 0$ whereas it has full rank. In \cite{FazayelliBanerjee16}, the authors show that $\cMGIG_d(\nu,\, A,\, B)$ is a unimodal distribution of which mode $M \in \dS_{++}^{\, d}$ is the unique solution of the algebraic Riccati equation $(d+1-2\, \nu)\, M + M B M = A$, and a standard importance sampling approach follows for the mean of the distribution. Our fallback solution is to solve this Riccati equation at each step and to replace all $\cMGIG_d$ random variables by the (unique) mode of the consecutive distributions. To assess the credibility of this \textit{ad hoc} sampling, the `oracle' models in which $\Oy$ and the shrinkage parameters are known are added to the simulations. We will see that, despite an unavoidable loss, the results remain pretty consistent. In particular, the support recovery does not appear to be impacted.
\end{rem}

\subsection{A simulation study}
\label{SecEmpSim}

In this section, the matrix of order $d \geq 1$ given by
\begin{equation*}
C_d = \big( \rho^{\vert i-j \vert} \big)_{1\, \leq\, i, j\, \leq\, d}
\end{equation*}
will be used as a typical covariance structure, for some $0 \leq \rho < 1$. Thus, the precision matrices will be chosen as a multiple of $C_d^{\, -1}$ to keep the same guideline in our simulations. The responses
\begin{equation*}
Y_k = B^{\, t}\, X_k + E_k
\end{equation*}
are generated through relations \eqref{Reparam} where, for all $1 \leq k \leq n$, $E_k \sim \cN(0, R)$. Because our models assume prior independence (or group-independence) in the columns of $\D$, it seems necessary to look at the influence of correlation among the predictors. So the standard choice $X_{k} \sim \cN(0, I_p)$ is first considered, but in some cases we will also test $X_{k} \sim \cN(0, C_p)$ for $\rho=0.5$ and $\rho=0.9$ to introduce a significant correlation between close predictors (see Figure \ref{FigMSPE}). For each experiment, the support recovery of $\D$ is evaluated thanks to the so-called $F$-score given by
\begin{equation*}
F = \frac{2\, p_{r}\, r_{e}}{p_{r}+r_{e}} \hsp \text{where} \hsp p_{r} = \frac{\textnormal{TP}}{\textnormal{TP} + \textnormal{FP}} \hsp \text{and} \hsp r_{e} = \frac{\textnormal{TP}}{\textnormal{TP} + \textnormal{FN}}
\end{equation*}
are the precision and the recall, respectively, and where T/F and P/N stand for true/false and positive/negative. To assess prediction skills, $n_e$ randomly chosen observations are used for estimation (for different $n_e$) and the remaining $n_v = n-n_e=100$ independent observations serve to compute the mean squared prediction error (MSPE). The results are compared to the ones obtained \textit{via} the penalized maximum of likelihood (PML) approach of \cite{YuanZhang14} thanks to the correctly adapted implementations of \cite{ChiquetEtAl17} and \cite{OkomeEtAl21}, with a cross-validated tuning parameter. In addition, we compute the sparse precision matrix estimations given by the graphical Lasso (GLasso) of \cite{FriedmanEtAl08}, and by the CLIME algorithm of \cite{CaiEtAl11}, using the \texttt{R} packages \texttt{glasso} and \texttt{fastclime}, respectively. Note that we always keep a small value for $q$, so $\D$ is penalized but not $\Oy$ when possible (PML and GLasso). Finally, the recent approach of \cite{RenEtAl15}, called ANT and based on the individual estimations of the partial correlations, is also implemented. Unlike PML, GLasso and CLIME, sparsity is not the result of penalizations for ANT but, instead, a threshold is deduced from the asymptotic normality of the estimates to decide which are significant and which can be set to zero. Let us add some preliminary comments about the methods compared in these simulations, all related to high-dimensional precision matrix estimation.
\begin{itemize}[label=$-$, leftmargin=*]
\item There is a important advantage in favor of our Bayesian approaches, PML and ANT because they do not need the estimation of $\Ox \in \dS_{++}^{\, p}$. Indeed, extracting the estimation of $\D \in \dR^{q \times p}$ and $\Oy \in \dS_{++}^{\, q}$ from that of the full precision matrix $\Omega \in \dS_{++}^{\, q+p}$ may generate a drastic bias when $p \gg q$, and that explains in particular why GLasso and CLIME give pretty bad results in what follows.
\item In its standard version, ANT is not designed to produce column-sparsity or group-sparsity in $\D$. So, by considering multiple testing at the column or even group level, we allow groups of coefficients to be zeroed simultaneously. We have observed that this modified ANT method (called ANT* in the simulations) loses a bit in prediction quality but is greatly  improved for support recovery.
\item Unfortunately, this is not appropriate for PML, GLasso and CLIME. It is therefore not surprising that they are largely outperformed by our Bayesian models and ANT* for (gs) and (sgs). Using group-penalties, which to the best of our knowledge still does not exist, should improve the results of these methods to some extent.
\end{itemize}

The seven scenarios below, from Scen. 0 to Scen. 6, as heterogeneous as possible, represent the diversity of the situations (high-dimensionality, kind of sparsity, dimension of the responses, coefficients hard to detect, etc.). We repeat each one $N=100$ or $N=50$ times, depending on the computation times involved, and the numerical results for $n_e=400$ and uncorrelated predictors are summarized in Table \ref{TabMSPE}. In addition, the evolution of MSPE is represented on Figure \ref{FigMSPE} for Scen. 1, 3 and 5, when $n_e$ grows from 100 to 500, both for uncorrelated and correlated predictors.
\begin{itemize}[label=$-$, leftmargin=*]
\item \textit{Scenario 0 (small dimension, no sparsity).} Let $q=1$, $p=5$ and set $\oy=1$. We fill $\Delta$ with $\cN(0, 2\, \oy)$ coefficients.
\item \textit{Scenario 1 (sparse direct links, univariate responses).} Let $q=1$, $p=50$ and set $\oy=1$. We randomly choose 10 locations of $\Delta$ filled with $\cN(0, \oy)$ coefficients while the others are zero.
\item \textit{Scenario 2 (sparse direct links, multivariate responses).} Let $q=2$, $p=80$ and set $\Oy = 2\, C_2^{\, -1}$ with $\rho=0.5$. We randomly choose 10 columns of $\Delta$ filled with $\cN_2(0, \Oy)$ coefficients while the others are zero.
\item \textit{Scenario 3 (group-sparse direct links, univariate responses).} Let $q=1$, $p=320$ and set $\oy=1$. We consider $m=5$ groups of size 100, 10, 100, 10 and 100. The two groups of size 10 are filled with $\cN(0, 0.5\, \oy)$ and $\cN(0, \oy)$ coefficients, respectively, while the other groups are zero.
\item \textit{Scenario 4 (group-sparse direct links, multivariate responses).} Let $q=3$, $p=500$ and set $\Oy = 3\, C_3^{\, -1}$ with $\rho=0.5$. We divide the columns of $\D$ into $m=25$ groups of size 20. We randomly choose 3 groups filled with $\cN_3(0, 0.5\, \Oy)$, $\cN_3(0, \Oy)$ and $\cN_3(0, 1.5\, \Oy)$ coefficients, respectively, while the other groups are zero.
\item \textit{Scenario 5 (sparse-group-sparse direct links, univariate responses).} Let $q=1$, $p=150$ and set $\oy=1$. We consider $m=3$ groups of size 50. Only the second group is non-zero, into which we randomly fill 10 locations with $\cN(0, \oy)$ coefficients.
\item \textit{Scenario 6 (sparse-group-sparse direct links, multivariate responses).} Let $q=5$, $p=1000$ and set $\Oy = 5\, C_5^{\, -1}$ with $\rho=0.5$. We divide the columns of $\D$ into $m=20$ groups of size 50, and a randomly chosen one is half filled with $\cN_5(0, \Oy)$ coefficients. The others columns of $\D$ are zero.
\end{itemize}

\begin{table}[!h]
\centering
\scriptsize
\begin{tabular}{c|c|c|c|c|c}
\hline
\hline
\multicolumn{6}{c}{Scenario 0} \\
\multicolumn{6}{c}{\tiny $\pi = 0$ \scriptsize} \\
\hline
\hline
Mod. & Shr. & MSPE & $F$ & $p_r$ & $r_e$ \\
\hline
(s-or) & - & 1.01 (0.11) & \underline{1.00} & 1.00 & 1.00 \\
(s) & (ad) & 1.03 (0.13) & \underline{1.00} & 1.00 & 1.00 \\
(s) & (gl) & 1.03 (0.13) & \underline{1.00} & 1.00 & 1.00 \\
PML & - & 1.01 (0.16) & \underline{1.00} & 1.00 & 1.00 \\
GLasso & - & \underline{1.00} (0.15) & \underline{1.00} & 1.00 & 1.00 \\
CLIME & - & \underline{1.00} (0.15) & \underline{1.00} & 1.00 & 1.00 \\
ANT* & - & 1.04 (0.13) & \underline{1.00} & 1.00 & 1.00 \\
\hline
\end{tabular}
~\\\medskip
\begin{tabular}{c|c|c|c|c|c}
\hline
\hline
\multicolumn{6}{c}{Scenario 1} \\
\multicolumn{6}{c}{\tiny $(a, b)=(p/2, 1)$ \scriptsize} \\
\hline
\hline
Mod. & Shr. & MSPE & $F$ & $p_r$ & $r_e$ \\
\hline
(s-or) & - & \underline{1.02} (0.13) & \underline{0.95} & 1.00 & 0.90 \\
(s) & (ad) & 1.04 (0.13) & \underline{0.95} & 1.00 & 0.90 \\
(s) & (gl) & 1.03 (0.13) & \underline{0.95} & 1.00 & 0.90 \\
PML & - & 1.08 (0.15) & 0.82 & 0.69 & 1.00 \\
GLasso & - & 2.37 (0.96) & 0.78 & 0.77 & 0.80 \\
CLIME & - & 2.52 (0.98) & 0.79 & 0.78 & 0.80 \\
ANT* & - & 1.25 (0.22) & 0.87 & 0.85 & 0.90 \\
\hline
\end{tabular}
\hsp
\begin{tabular}{c|c|c|c|c|c}
\hline
\hline
\multicolumn{6}{c}{Scenario 2} \\
\multicolumn{6}{c}{\tiny $(a, b)=(p, 1)$ \scriptsize} \\
\hline
\hline
Mod. & Shr. & MSPE & $F$ & $p_r$ & $r_e$ \\
\hline
(s-or) & - & \underline{0.52} (0.09) & \underline{0.95} & 1.00 & 0.90 \\
(s) & (ad) & 0.54 (0.09) & \underline{0.95} & 1.00 & 0.90 \\
(s) & (gl) & 0.55 (0.08) & \underline{0.95} & 1.00 & 0.90 \\
PML & - & 0.77 (0.15) & 0.86 & 1.00 & 0.75 \\
GLasso & - & 1.74 (0.49) & 0.72 & 0.91 & 0.60 \\
CLIME & - & 1.11 (0.35) & 0.73 & 0.76 & 0.70 \\
ANT* & - & 1.04 (0.44) & 0.90 & 0.89 & 0.91 \\
\hline
\end{tabular}
~\\\medskip
\begin{tabular}{c|c|c|c|c|c}
\hline
\hline
\multicolumn{6}{c}{Scenario 3} \\
\multicolumn{6}{c}{\tiny $(a, b)=(m, 1)$ \scriptsize} \\
\hline
\hline
Mod. & Shr. & MSPE & $F$ & $p_r$ & $r_e$ \\
\hline
(gs-or) & - & \underline{1.03} (0.27) & \underline{1.00} & 1.00 & 1.00 \\
(gs) & (ad) & 1.04 (0.27) & \underline{1.00} & 1.00 & 1.00 \\
(gs) & (gl) & 1.04 (0.34) & \underline{1.00} & 1.00 & 1.00 \\
PML & - & 1.80 (0.36) & 0.89 & 1.00 & 0.80 \\
GLasso & - & 4.23 (1.61) & 0.58 & 0.50 & 0.70 \\
CLIME & - & 2.98 (1.22) & 0.68 & 0.90 & 0.55 \\
ANT* & - & 1.52 (0.95) & \underline{1.00} & 1.00 & 1.00 \\
\hline
\end{tabular}
\hsp
\begin{tabular}{c|c|c|c|c|c}
\hline
\hline
\multicolumn{6}{c}{Scenario 4} \\
\multicolumn{6}{c}{\tiny $(a, b)=(m, 1)$ \scriptsize} \\
\hline
\hline
Mod. & Shr. & MSPE & $F$ & $p_r$ & $r_e$ \\
\hline
(gs-or) & - & \underline{0.40} (0.14) & \underline{1.00} & 1.00 & 1.00 \\
(gs) & (ad) & 0.45 (0.16) & \underline{1.00} & 1.00 & 1.00 \\
(gs) & (gl) & 0.46 (0.17) & \underline{1.00} & 1.00 & 1.00 \\
PML & - & 3.18 (0.53) & 0.75 & 0.94 & 0.62 \\
GLasso & - & 9.46 (1.38) & 0.46 & 0.66 & 0.35 \\
CLIME & - & 8.32 (1.51) & 0.48 & 0.45 & 0.52 \\
ANT* & - & 6.53 (1.22) & \underline{1.00} & 1.00 & 1.00 \\
\hline
\end{tabular}
~\\\medskip
\begin{tabular}{c|c|c|c|c|c}
\hline
\hline
\multicolumn{6}{c}{Scenario 5} \\
\multicolumn{6}{c}{\tiny $(a_1, b_1, a_2, b_2) = (m, 1, p/3, 1)$ \scriptsize} \\
\hline
\hline
Mod. & Shr. & MSPE & $F$ & $p_r$ & $r_e$ \\
\hline
(sgs-or) & - & \underline{1.00} (0.15) & \underline{0.96} & 1.00 & 0.92 \\
(sgs) & (ad) & 1.04 (0.16) & 0.95 & 1.00 & 0.91 \\
(sgs) & (gl) & 1.03 (0.16) & 0.91 & 1.00 & 0.84 \\
PML & - & 1.92 (0.60) & 0.89 & 1.00 & 0.80 \\
GLasso & - & 3.48 (1.30) & 0.78 & 0.86 & 0.71 \\
CLIME & - & 1.88 (0.92) & 0.79 & 1.00 & 0.65 \\
ANT* & - & 1.26 (0.98) & 0.88 & 0.86 & 0.90 \\
\hline
\end{tabular}
\hsp
\begin{tabular}{c|c|c|c|c|c}
\hline
\hline
\multicolumn{6}{c}{Scenario 6} \\
\multicolumn{6}{c}{\tiny $(a_1, b_1, a_2, b_2) = (m, 1, p/20, 1)$ \scriptsize} \\
\hline
\hline
Mod. & Shr. & MSPE & $F$ & $p_r$ & $r_e$ \\
\hline
(sgs-or) & - & \underline{0.21} (0.13) & \underline{1.00} & 1.00 & 1.00 \\
(sgs) & (ad) & 0.24 (0.32) & \underline{1.00} & 1.00 & 1.00 \\
(sgs) & (gl) & 0.24 (0.33) & \underline{1.00} & 1.00 & 1.00 \\
PML & - & 0.50 (0.17) & 0.83 & 0.95 & 0.74 \\
GLasso & - & 3.83 (0.77) & 0.50 & 0.97 & 0.34 \\
CLIME & - & 2.98 (0.51) & 0.51 & 1.00 & 0.34 \\
ANT* & - & 2.10 (0.72) & \underline{1.00} & 1.00 & 1.00 \\
\hline
\end{tabular}
\normalsize
\vspace{0.3cm}
\caption{Medians of the mean squared prediction errors (with standard deviations), $F$-scores, precisions and recalls after $N=100$ repetitions of Scen. 0 to Scen. 6 ($N=50$ for Scen. 4 and Scen. 6), with $n_e=400$ and uncorrelated predictors. The hyperparameters chosen for the prior spike probability are also indicated.}
\label{TabMSPE}
\end{table}

\begin{figure}[!h]
\centering
\includegraphics[width=10cm]{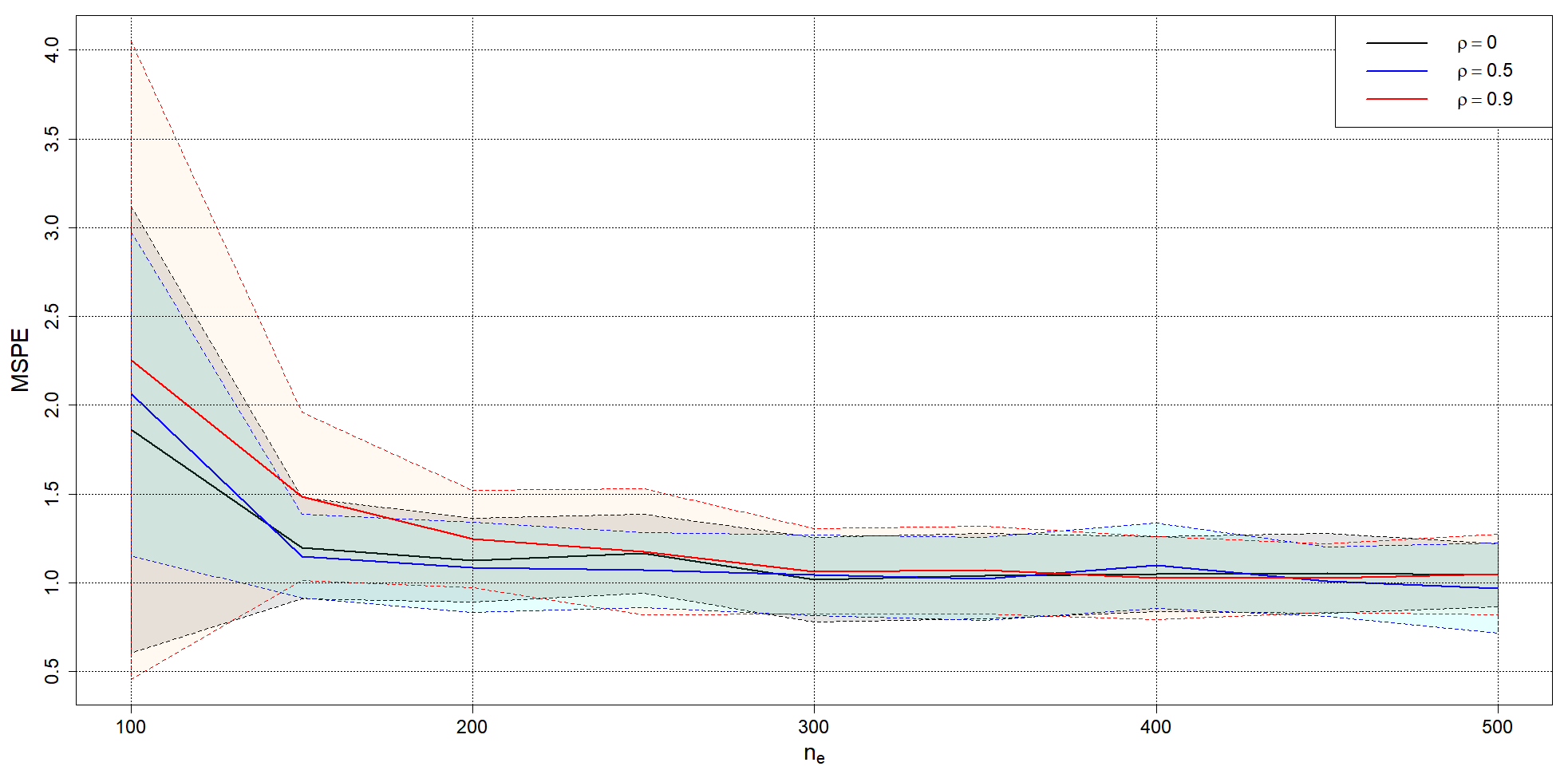} \includegraphics[width=10cm]{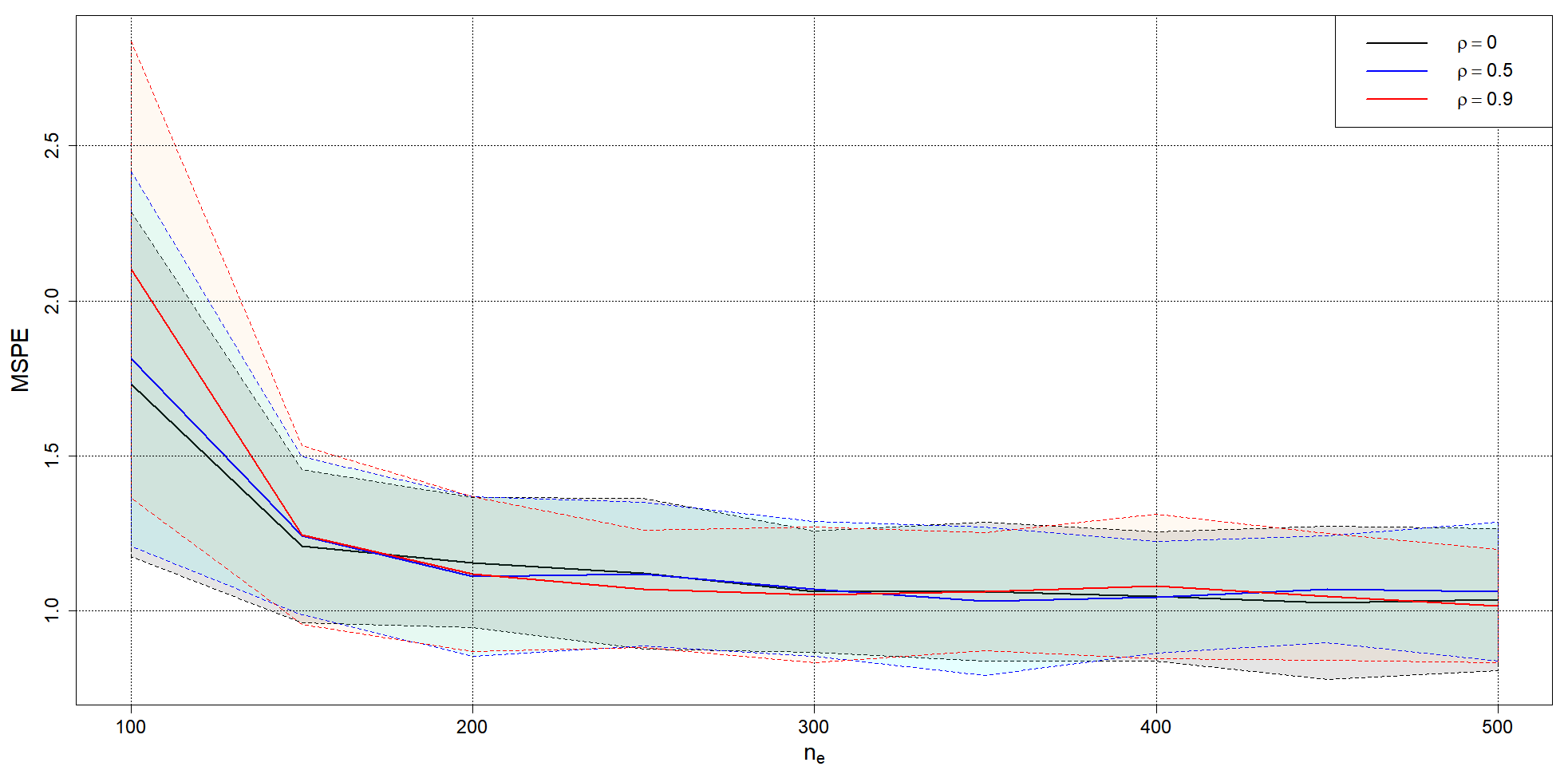} \includegraphics[width=10cm]{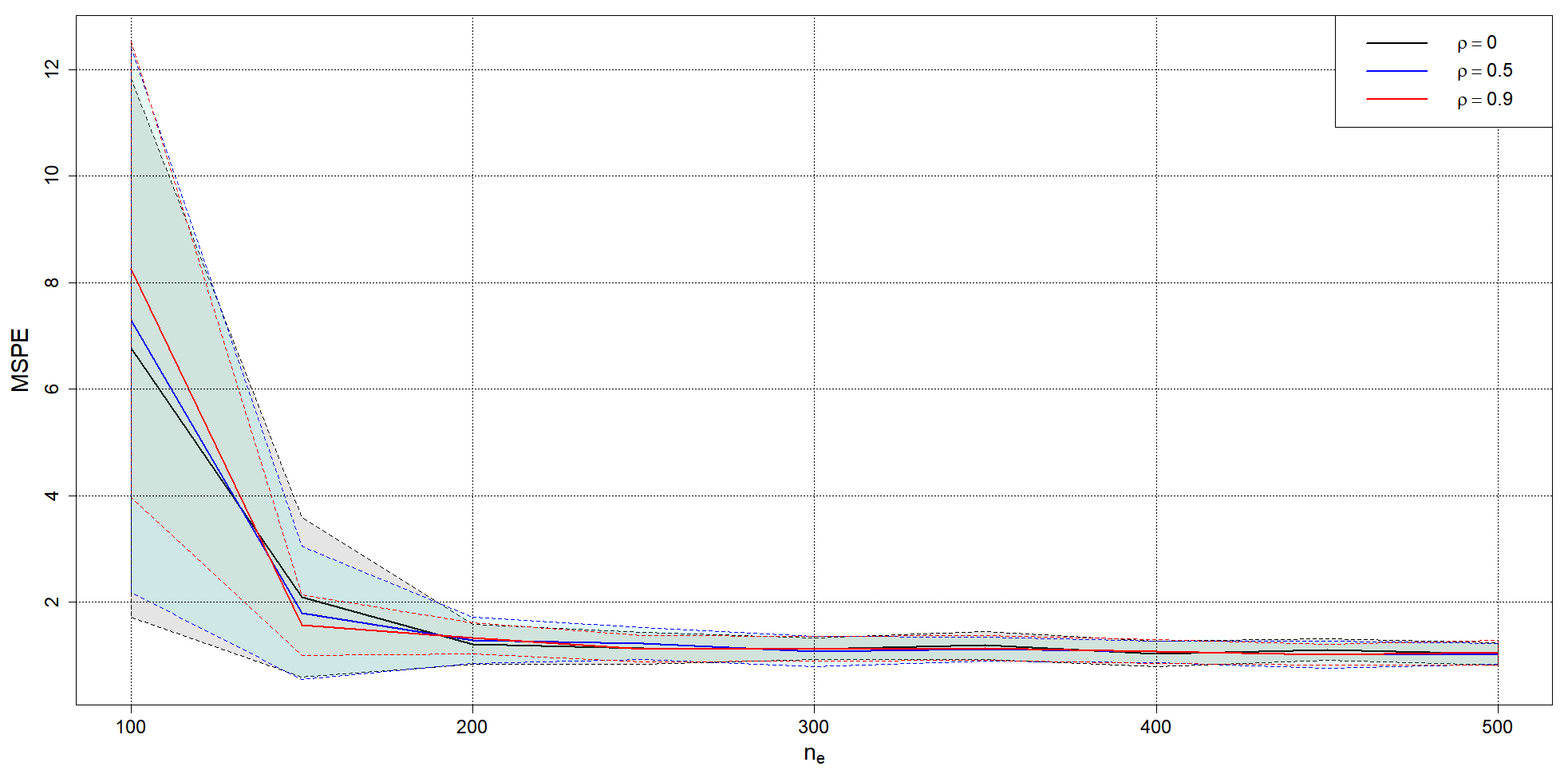}
\caption{Medians of the mean squared prediction errors obtained after $N=100$ repetitions of Scen. 1 (top), Scen. 3 (middle) and Scen. 5 (bottom) with $\pm 1$ standard deviation and $n_{e}$ growing from 100 to 500. The black curves correspond to uncorrelated predictors ($\rho=0$) while the blue and red curves correspond to correlated predictors ($\rho=0.5$ and $\rho=0.9$, respectively).}
\label{FigMSPE}
\end{figure}

Now, let us try to summarize our observations. In terms of support recovery, the Bayesian spike-and-slab framework and the modified ANT* method give results incomparably better than the sparsity-inducing penalized approaches (PML, GLasso and CLIME). As suggested in Rem. 3.3 of \cite{OkomeEtAl21}, this may be a consequence of the fact that the cross-validation steps calibrate the models to reach the best prediction error, sometimes at the cost of support recovery by picking a small penalty level. The superiority of ANT over GLasso and CLIME is recognized and discussed in \cite{RenEtAl15}, but this also highlights the ability of our Bayesian models to reach good results both in prediction and in support recovery. However as soon as $p$ is large (say, $p \geq 500$), the Bayesian studies should be conducted with a group structure or by promoting very sparse settings. Otherwise, due to the outline of the sampler, looping over each column of $\D$ may quickly become intractable. But, even when $p \gg n$ (see Scen. 4 and 6), the Bayesian grouped models give much better results than the penalized ones. The procedures are obviously very sensitive to the initialization of the sampler, especially when $p \gg n$. For example, the term $\vert I_{\kappa_g} + \lambda_g\, \gdX_g^{\, t}\, \gdX_g \vert$ is likely to explode when $\kappa_g$ is large and $\lambda_g > 1$, that is why $\lambda_g$ has to be carefully controlled \textit{via} an accurate initial choice of $\ell_g$. As long as we avoid those numerical pitfalls, the Bayesian models generally give very stable results (see \textit{e.g.} the standard deviations related to the MSPEs) which is probably due to the good performances in support recovery. More precisely, Figure \ref{FigMSPE} shows that this stability in the results is reached from $n_e=200$ observations in the learning set: for $n_e < 200$ the MSPEs are rather chaotic before stabilizing. The same figure also highlights that the presence of correlation in the predictors do not seem to have a significant effect on the estimation procedure, except for small size samples and high correlation where the degradation is noticeable. Unfortunately, despite all our efforts to optimize the Gibbs samplers, our procedures cannot compete with the Lasso-type algorithms (GLasso, CLIME or even ANT) in terms of computational times. This is an issue on which future studies should focus (ongoing works are devoted to translating the samplers into more efficient environments). Overall, the real strenght of the Bayesian spike-and-slab approach is clearly the support recovery of the direct links between predictors and responses but it seems that one can hardly expect to deal with very high-dimensional studies as long as we do not impose a group structure or a huge degree of sparsity. The highly competitive MSPEs obtained confirm the relevance of Bayesian PGGMs not only for variables selection but also for prediction purposes in the context of high-dimensional regressions.

\subsection{A real dataset}
\label{SecEmpReal}

Let us now study the \texttt{Hopx} dataset, fully described in \cite{PetrettoEtAl10}. It contains $p=770$ genetic markers spread over $m=20$ chromosomes from $n=29$ inbred rats. It also contains the corresponding measured gene expression levels of $q=4$ tissues (adrenal gland, fat, heart and kidney). The goal is to identity a sparse set of predictors that jointly explain the outcomes, with the natural group structure formed by chromosomes (see Table \ref{TabChr}). This dataset has already been analyzed in \cite{LiquetEtAl16}, using a Bayesian regression without group structure, and later in \cite{LiquetEtAl17} including group and sparse-group structures. So the PGGM is supposed to bring new perspectives about relationships in terms of partial correlations. A particularity of this dataset is that the responses are very correlated, so we should expect an estimation of $\Oy^{-1}$ with significant non-diagonal elements and a clear advantage in using PGGMs. Indeed, a predictor considered to be influencing all the outcomes could be the result of a direct relation to one tissue propagated to the others by an artificial correlation effect. As can be seen on Figure \ref{FigCorr}, the predictors are also highly correlated with their neighbors (for the sake of readability, we only represent the correlogram of predictors located on chromosomes 8, 9 and 10).

\begin{table}[!h]
\centering
\scriptsize
\begin{tabular}{|c|c|c|c|c|c|c|c|c|c|c|c|c|c|c|c|c|c|c|c|c|}
\hline
Chr. & 1 & 2 & 3 & 4 & 5 & 6 & 7 & 8 & 9 & 10 & 11 & 12 & 13 & 14 & 15 & 16 & 17 & 18 & 19 & 20 \\
\hline
Nb. & 74 & 67 & 63 & 60 & 39 & 45 & 52 & 43 & 31 & 51 & 21 & 26 & 33 & 22 & 15 & 27 & 18 & 30 & 34 & 19 \\
\hline
\end{tabular}
\normalsize
\vspace{0.3cm}
\caption{Number of markers on each chromosome, which correspond to the sizes $\kappa_g$ of each group for $1 \leq g \leq 20$ when running (gs) and (sgs).}
\label{TabChr}
\end{table}

\begin{figure}[!h]
\centering
\includegraphics[width=7.5cm]{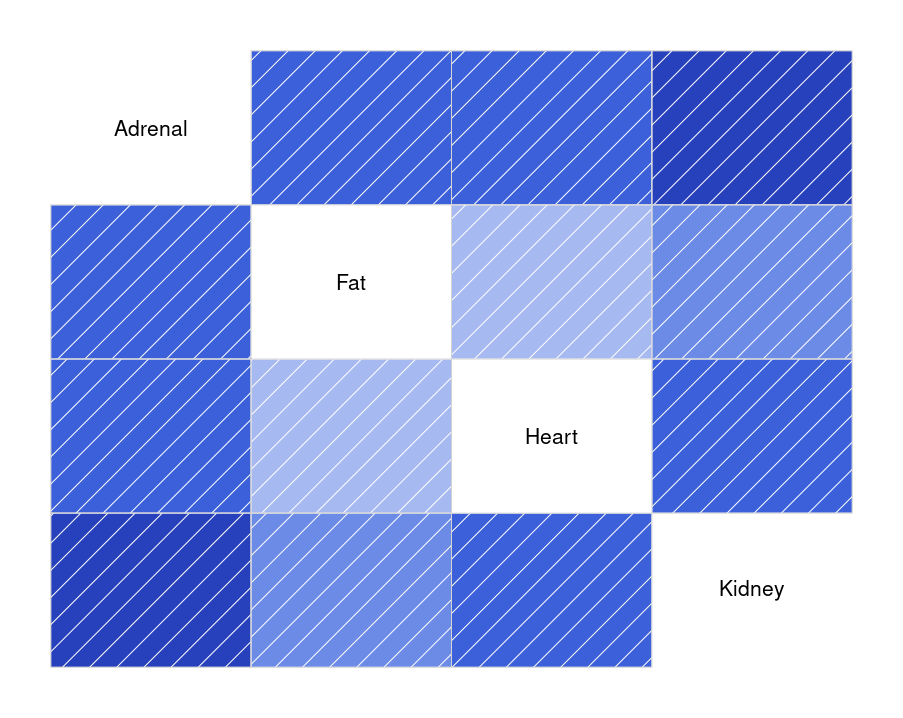} \includegraphics[width=7.5cm]{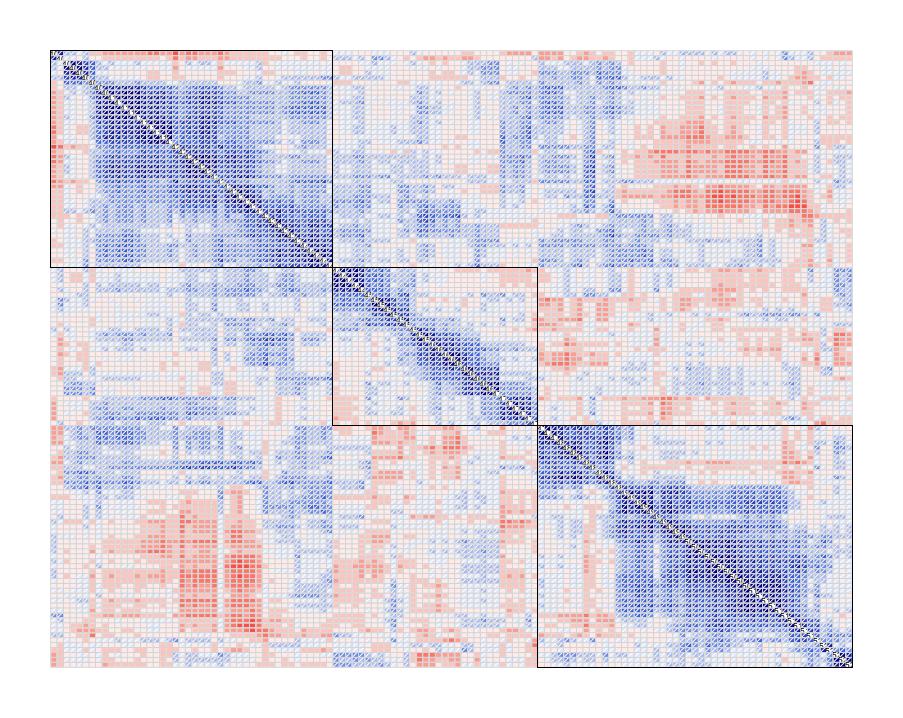}
\caption{Correlogram of responses (left) and correlogram of predictors located on chromosomes 8, 9 and 10 (right). The colormap associates red with negative correlations and blue with positive correlations.}
\label{FigCorr}
\end{figure}

The small sample size relative to the number of covariates ($29/770$) weakens the study. To strengthen our conclusions, we decided to run $N=100$ experiments based on 25 randomly chosen observations and to aggregate the results. We first investigate the selection of predictors at the chromosomes scale, \textit{i.e.} we run (gs) according to the previous protocol with an adaptative shrinkage. Due to the high-dimensionality, very sparse models are naturally selected at the groups scale so we `help' the selection process by choosing $a=1$ and $b=m$ in the prior probability $\pi$. The empirical distribution of the posterior probability of inclusion for each chromosome is represented on the left of Figure \ref{FigProbG}. The selection procedure focuses on chromosomes 14, 15 and 17 (and not just on chromosomes 2 and 3 as in \cite{LiquetEtAl17}) but the estimation process gives an overwhelming advantage to chromosome 14, far ahead of its neighbors. This is undoubtedly the influence of \texttt{D14Mit3}, a marker located on chromosome 14 and known to have a very significant effect on this dataset. The main conclusion to be drawn at this stage is that chromosome 14 has a positive effect on \texttt{Fat} and a \textit{negative} effect on \texttt{Heart}, as can also be seen on the right of Figure \ref{FigProbG}. Therefore, it is likely that the overall positive influence of \texttt{D14Mit3} identified by previous authors is due to the combination of a direct positive link with \texttt{Fat}, a direct negative link with \texttt{Heart} and a correlation effect from the outcomes. This hypothesis is given additional credibility by the numerical results: from (gs), the corresponding column of $\D$ is approximately $(0.00,\, 0.04, -0.09,\, 0.00)$ which, through relations \eqref{Reparam}, leads to $(0.15,\, 0.25,\, 0.34,\, 0.21)$ as estimated regression coefficients. This roughly corresponds to the values indicated in Tab. 2 of \cite{LiquetEtAl17}, at least for the main effect on \texttt{Heart}. Thus for chromosome 14, the numerical results coincide but the interpretations are clearly different. Of course, similar reasonings can be carried out for the less influent chromosomes.

\begin{figure}[!h]
\centering
\includegraphics[width=9cm, height=5cm]{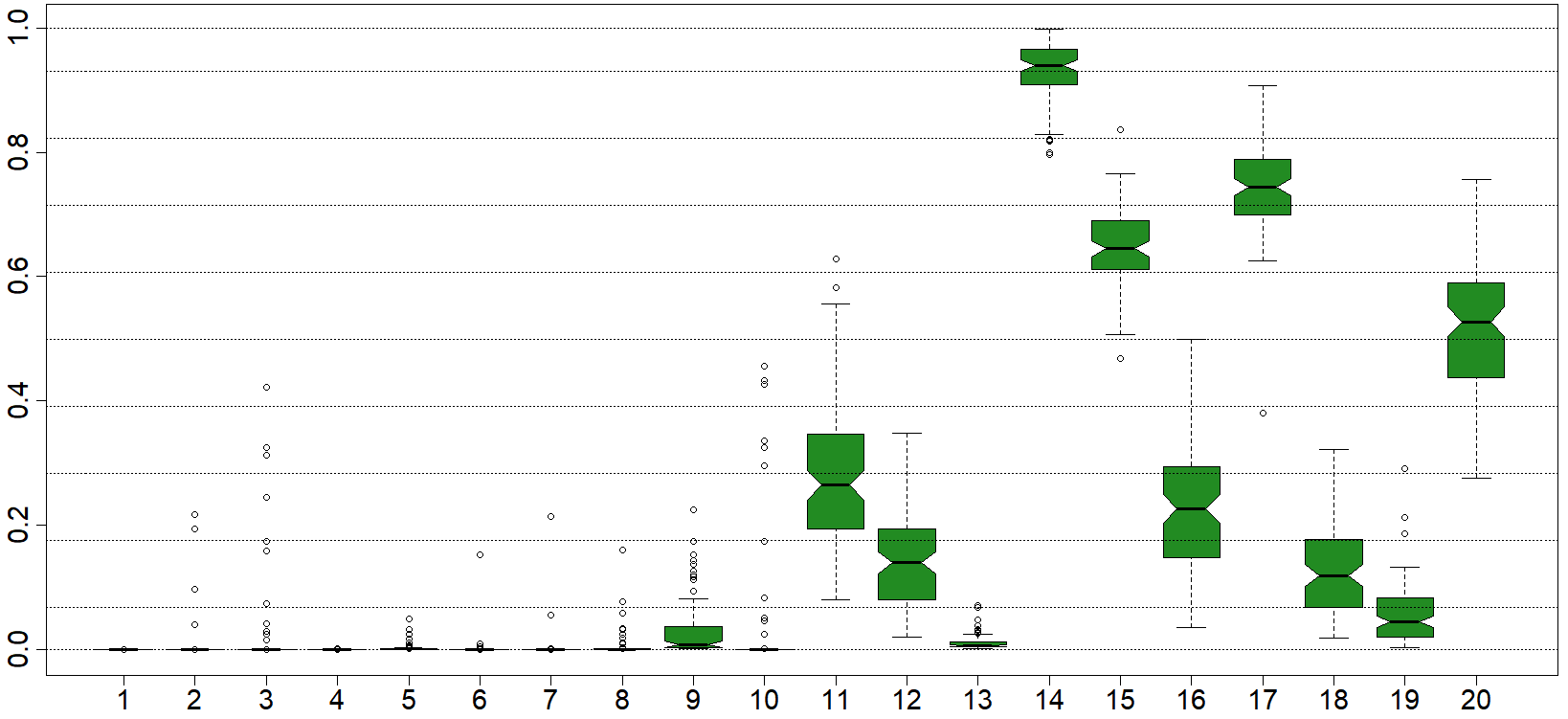} \hsp \includegraphics[width=5cm, height=5.1cm]{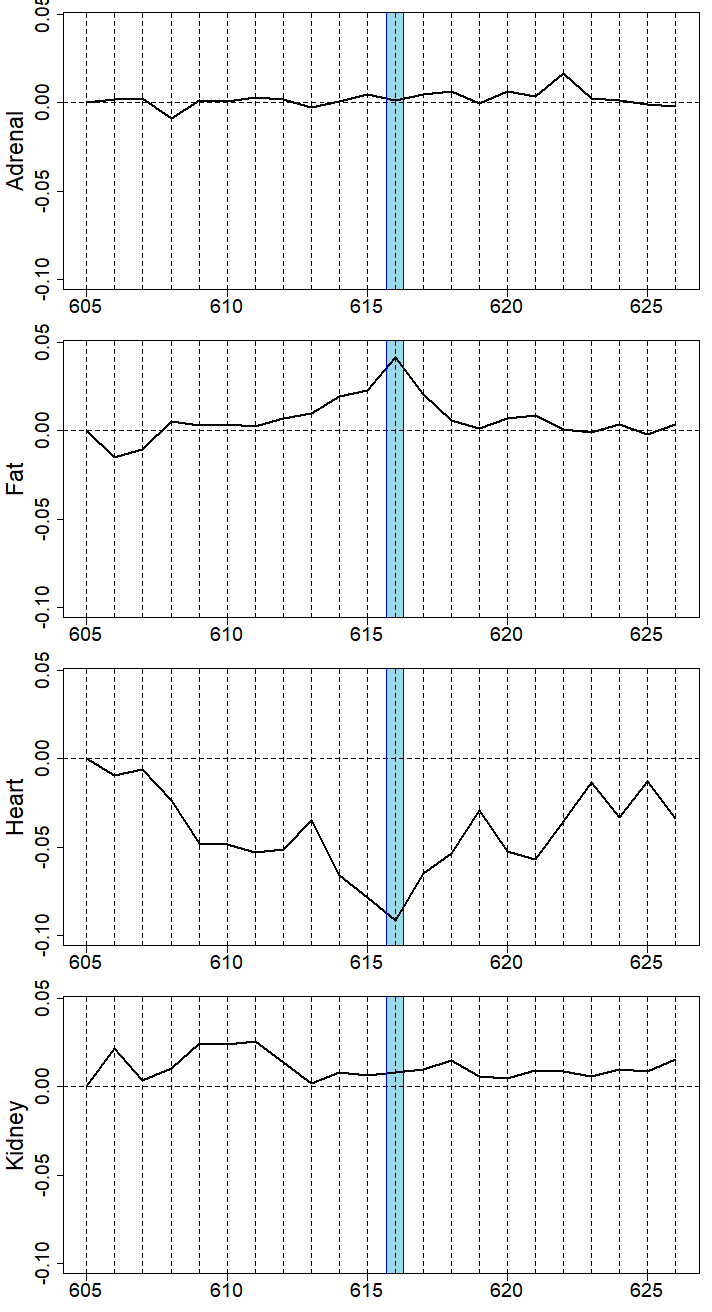}
\caption{Empirical distribution of the posterior probability of inclusion estimated by (gs) for each chromosome (left). Aggregated (gs) estimation of $\Delta$ on chromosome 14 with \texttt{D14Mit3} hilighted (right).}
\label{FigProbG}
\end{figure}

It is perhaps more interesting to look for a bi-level selection in order to identify a sparse set of markers and not only chromosomes. In this regard, (sgs) is launched using the same statistical protocol, adaptative shrinkage and hardly informative hypermarameters $a_1=3$, $b_1=1$, $a_2=1$ and $b_2=1$ which happen to be sufficient to generate a huge degree of sparsity. While many chromosomes are excluded from the model given by (gs), with (sgs) we see some contributions localized in certain chromosomes having little influence when taken as a whole. At the markers scale, the randomness of the sampler and the high level of correlation between close predictors probably explain the presence of artifacts which sometimes make it difficult to distinguish the real contributions from the background noise. We therefore use the $N=100$ experiments to build $95\%$ confidence intervals and keep only significant estimates. By way of example, Figure \ref{FigEstSGS} displays the results obtained on chromosomes 7, 8 and 14. The main markers standing out are summarized in Table \ref{TabChrSGS} together with the kind of direct influences detected. Markers already highlighted in \cite{LiquetEtAl16} or \cite{LiquetEtAl17} are also indicated. One can see that most of our conclusions coincide, but new markers are suggested (especially on chromosome 8) and others have disappeared. Overall, the more stringent statistical protocol that we used led to the retention of fewer predictors with more guarantee. An important consequence of this study is the new interpretations in terms of direct influences allowed by PGGMs. Especially as the residual correlations, hidden in the estimation of $R = \Oyi$ and closely related to the correlations between the responses, are very high (greater than $0.7$), as we suspected from Figure \ref{FigCorr}.

\begin{figure}[!h]
\centering
\includegraphics[width=5cm, height=5.1cm]{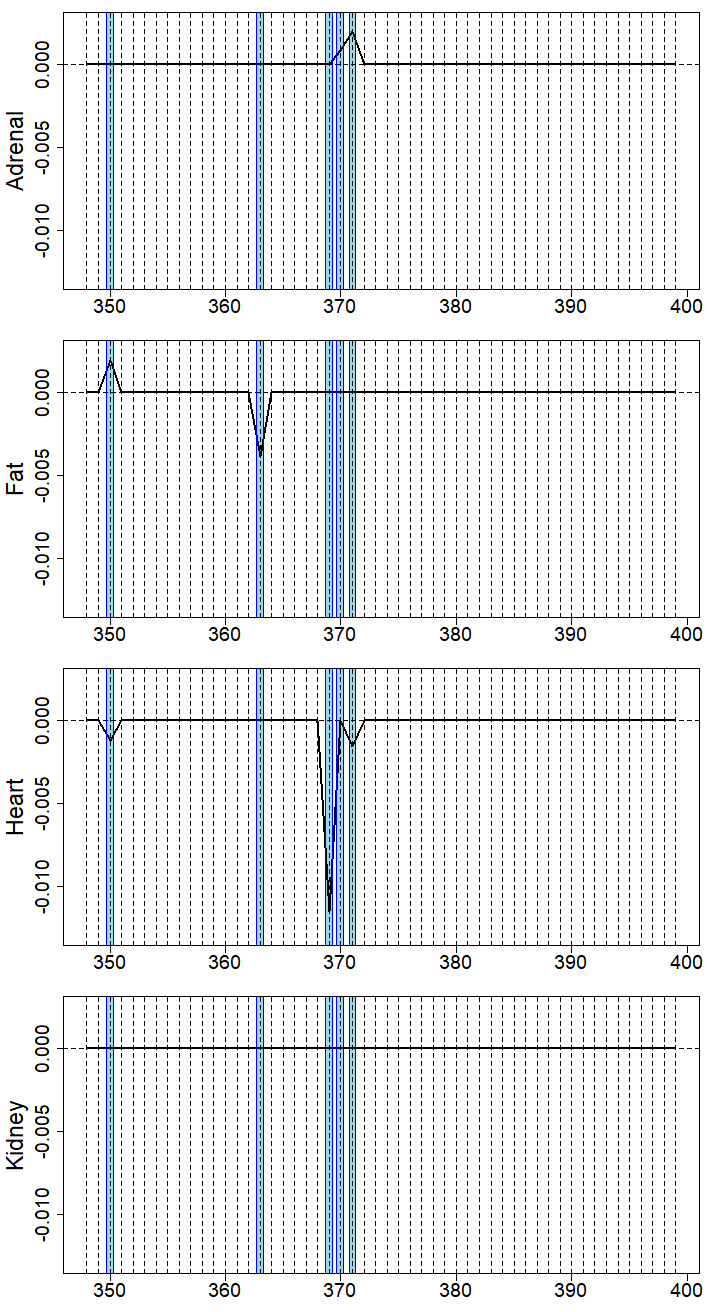} ~ \includegraphics[width=5cm, height=5.1cm]{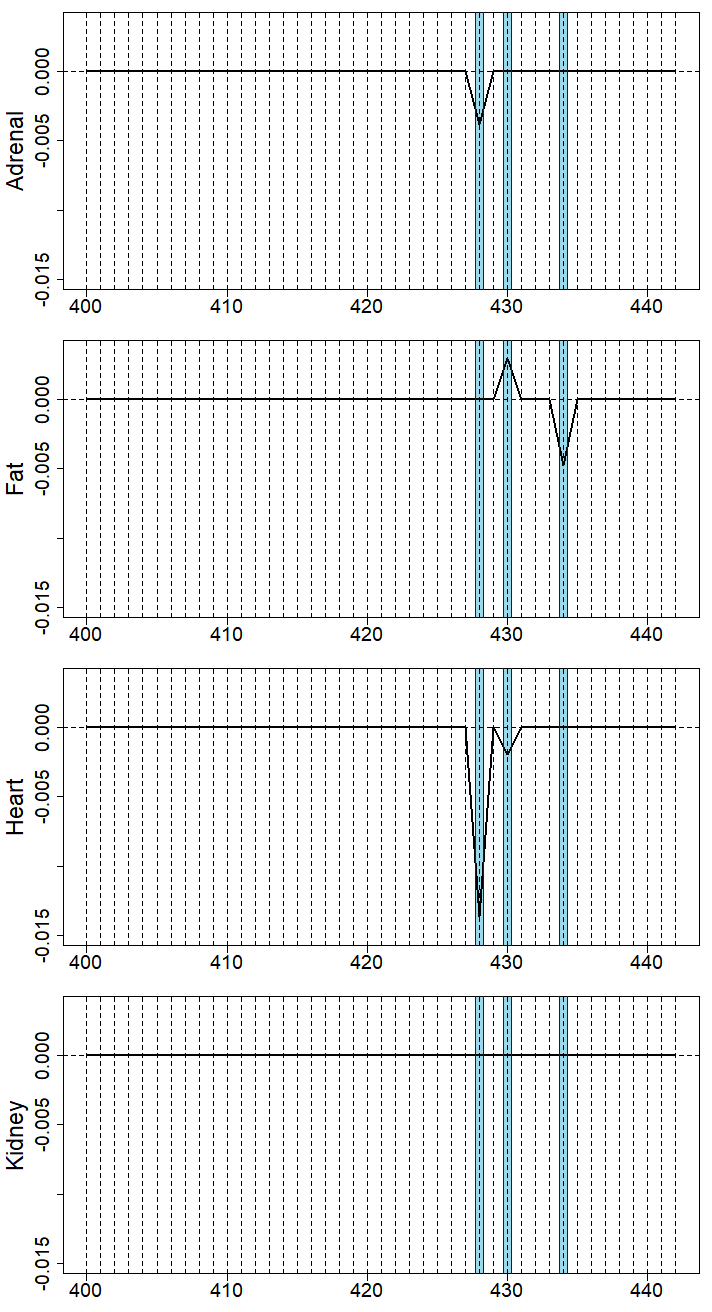} ~ \includegraphics[width=5cm, height=5.1cm]{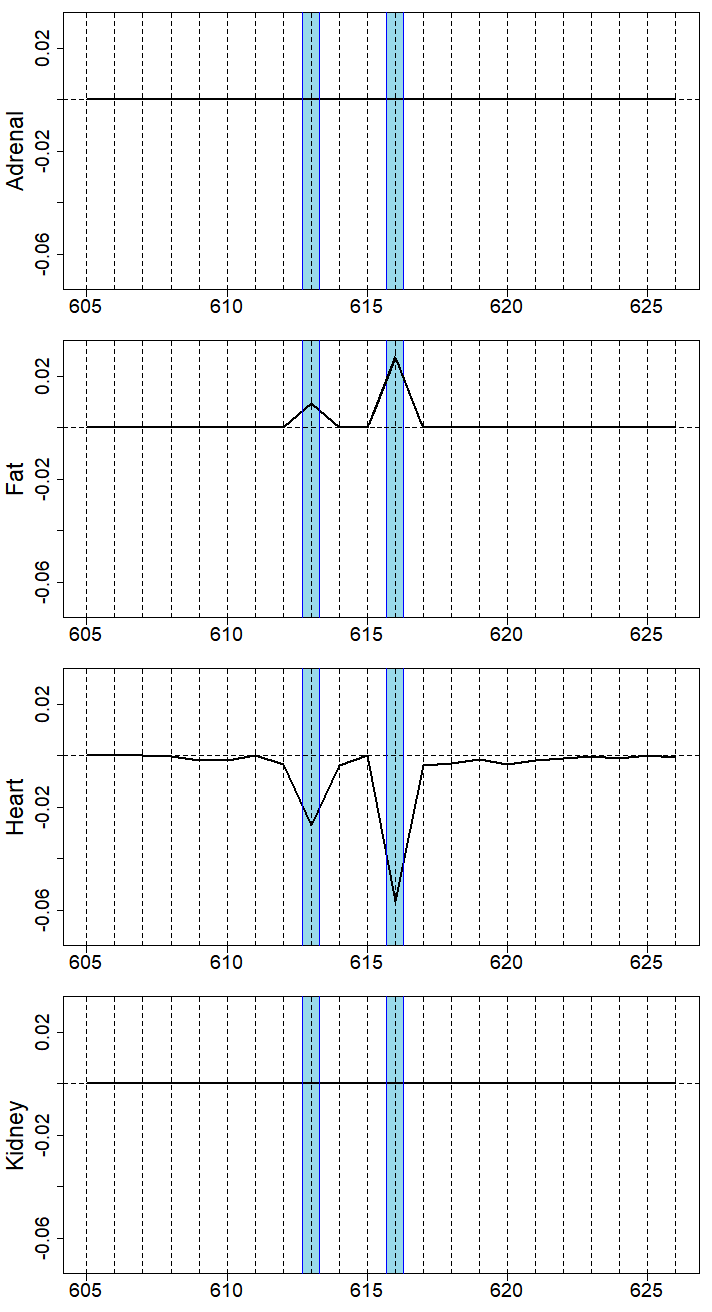}
\caption{Aggregated (sgs) estimation of $\Delta$ on chromosomes 7, 8 and 14, from left to right. The hilighted markers are \texttt{D7Cebr205s3}, \texttt{D7Mit6}, \texttt{D7Rat19}, \texttt{Myc} and \texttt{D7Rat17} for chromosome 7, \texttt{D8Mgh4}, \texttt{D8Rat135} and \texttt{Rbp2} for chromosome 8 and \texttt{D14Rat8} and \texttt{D14Mit3} for chromosome 14.}
\label{FigEstSGS}
\end{figure}

\begin{table}[!h]
\centering
\scriptsize
\begin{tabular}{|c|c|c|}
\hline
Chromosomes & Markers & Main direct influences \\
\hline
3 & \texttt{D3Mit16}* & \texttt{Adrenal}+ \texttt{Heart}-- \\
\hline
\multirow{5}{*}{7} & \texttt{D7Cebr205s3}* & \texttt{Fat}+ \texttt{Heart}-- \\
\cline{2-3}
 & \texttt{D7Mit6}* & \texttt{Fat}-- \\
\cline{2-3}
 & \texttt{D7Rat19}* & \texttt{Heart}-- \\
\cline{2-3}
 & \texttt{Myc}* & \texttt{Adrenal}+ \\
\cline{2-3}
 & \texttt{D7Rat17} & \texttt{Adrenal}+ \texttt{Heart}-- \\
\hline
\multirow{3}{*}{8} & \texttt{D8Mgh4} & \texttt{Adrenal}-- \texttt{Heart}-- \\
\cline{2-3}
 & \texttt{D8Rat135} & \texttt{Fat}+ \texttt{Heart}-- \\
\cline{2-3}
 & \texttt{Rbp2} & \texttt{Fat}-- \\
\hline
\multirow{3}{*}{10} & \texttt{D10Rat33}* & \texttt{Adrenal}+ \\
\cline{2-3}
 & \texttt{D10Mit3}* & \texttt{Adrenal}+ \\
\cline{2-3}
 & \texttt{D10Rat31}* & \texttt{Fat}-- \\
\hline
11 & \texttt{D11Rat47} & \texttt{Fat}-- \\
\hline
\multirow{2}{*}{14} & \texttt{D14Rat8}* & \texttt{Fat}+ \texttt{Heart}-- \\
\cline{2-3}
& \texttt{D14Mit3}* & \texttt{Fat}+ \texttt{Heart}-- \\
\hline
\multirow{2}{*}{15} & \texttt{D15Cebr7s13} & \texttt{Kidney}-- \\
\cline{2-3}
 & \texttt{D15Rat21}* & \texttt{Adrenal}+ \texttt{Kidney}-- \\
\hline
17 & \texttt{Prl} & \texttt{Adrenal}-- \texttt{Kidney}-- \\
\hline
20 & \texttt{D20Rat55} & \texttt{Kidney}-- \\
\hline
\end{tabular}
\normalsize
\vspace{0.3cm}
\caption{Main relations detected by (sgs). \texttt{X}* means that marker \texttt{X} has already been suggested by previous authors in this dataset. \texttt{Y}-- (\texttt{Y}+) means that response \texttt{Y} is negatively (positively) influenced by \texttt{X}.}
\label{TabChrSGS}
\end{table}

To conclude, we would like to draw the attention of the reader to some weaknesses of the study, already mentioned in Section \ref{SecEmpSim} and still under investigation. On the one hand, the sampler must be carefully initialized to avoid numerical issues, particularly for group structures. For example, our heuristic approach is to initialize $\ell_g$ such that $\dE[\lambda_g] < 1$ to control the behavior of $\vert I_{\kappa_g} + \lambda_g\, \gdX_g^{\, t}\, \gdX_g \vert$ during the first iterations. This works pretty well in practice, but needs to be done on a case-by-case basis, which could be improved. On the other hand, according to our experiments, the very high-dimensional studies (say, $p > 10^4$) cannot be conducted in a reasonable time thanks to our samplers, except by enforcing a very huge degree of sparsity. Enhanced MCMC methods may be useful or novel computational strategies like the `shotgun' stochastic algorithm of \cite{YangNarisetty20}. From a theoretical point of view, we should obviously improve the estimation procedure by sampling from the $\cMGIG_q$ distribution for $q > 1$, and not using the mode. Our fallback solution gives satisfactory but not completely rigorous results. In addition, it could be interesting to generalize the support recovery guarantee of Proposition \ref{PropSupRecGS} to (sgs), which is certainly possible at the cost of a few additional developments. Overall, our study shows that for the moderate values of $p$ (up to $10^3$ or $10^4$), the Bayesian approach of the partial Gaussian graphical models is a very serious alternative to the frequentist penalized estimations, for prediction but also and especially for support recovery.

\bigskip

\noindent \textbf{Acknowledgements and Fundings.} The authors thank ALM (Angers Loire M\'etropole) and the ICO (Institut de Canc\'erologie de l'Ouest) for the financial support. This work is partially financed through the ALM grant and the ``Programme op\'erationnel r\'egional FEDER-FSE Pays de la Loire 2014-2020" noPL0015129 (EPICURE). The authors also thank Mario Campone (project leader and director of the ICO), Mathilde Colombi\'e (scientific coordinator of EPICURE clinical trial) and Fadwa Ben Azzouz, biomathematician in Bioinfomics, for the initiation, the coordination and the smooth running of the project.

\nocite{*}

\bibliographystyle{acm}
\bibliography{PGGM_Bayes}

\begin{thebibliography}{10}

\bibitem{BaiEtAl20}
{\sc Bai, R., Moran, G.~E., Antonelli, J.~L., Chen, Y., and Boland, M.~R.}
\newblock Spike-and-slab group lassos for grouped regression and sparse
  generalized additive models.
\newblock {\em J. Am. Stat. Assoc.\/} (2020), 1--14.

\bibitem{BanerjeeEtAl08}
{\sc Banerjee, O., El~Ghaoui, L., and D'Aspremont, A.}
\newblock Model selection through sparse maximum likelihood estimation for
  multivariate {G}aussian or binary data.
\newblock {\em J. Mach. Learn. Res. 9\/} (2008), 485--516.

\bibitem{BrownEtAl98}
{\sc Brown, P.~J., Vannucci, M., and Fearn, T.}
\newblock Multivariate {B}ayesian variable selection and prediction.
\newblock {\em J. R. Statist. Soc. B. 60}, 3 (1998), 627--641.

\bibitem{CaiEtAl11}
{\sc Cai, T., Liu, W., and Luo, X.}
\newblock A constrained $\ell_1$ minimization approach to sparse precision
  matrix estimation.
\newblock {\em J. Am. Stat. Assoc. 106}, 494 (2011), 594--607.

\bibitem{CaiZhou12}
{\sc Cai, T., and Zhou, H.}
\newblock Optimal rates of convergence for sparse covariance matrix estimation.
\newblock {\em Ann. Stat. 40}, 5 (2012), 2389--2420.

\bibitem{ChiquetEtAl17}
{\sc Chiquet, J., Mary-Huard, T., and Robin, S.}
\newblock Structured regularization for conditional {G}aussian graphical
  models.
\newblock {\em Stat. Comput. 27}, 3 (2017), 789--804.

\bibitem{EltoftEtAl06}
{\sc Eltoft, T., Kim, T., and Lee, T.}
\newblock Multivariate scale mixture of {G}aussians modeling.
\newblock In {\em Independent Component Analysis and Blind Signal Separation\/}
  (2006), Springer Berlin Heidelberg, pp.~799--806.

\bibitem{FangEtAl20}
{\sc Fang, Y., Karlis, D., and Subedi, S.}
\newblock A {B}ayesian approach for clustering skewed data using mixtures of
  multivariate normal-inverse {G}aussian distributions.
\newblock {\em arXiv:2005.02585\/} (2020).

\bibitem{FazayelliBanerjee16}
{\sc Fazayeli, F., and Banerjee, A.}
\newblock {\em The {M}atrix {G}eneralized {I}nverse {G}aussian distribution:
  properties and applications}, vol.~9851 of {\em Frasconi P., Landwehr N.,
  Manco G., Vreeken J. (eds) Machine Learning and Knowledge Discovery in
  Databases. ECML PKDD 2016. Lecture Notes in Computer Science}.
\newblock Springer, Cham., 2016.

\bibitem{FriedmanEtAl08}
{\sc Friedman, J., Hastie, T., and Tibshirani, R.}
\newblock Sparse inverse covariance estimation with the graphical {L}asso.
\newblock {\em Biostatistics. 9}, 3 (2008), 432--441.

\bibitem{GanEtAl19}
{\sc Gan, L., Yang, X., Narisetty, N., and Liang, F.}
\newblock Bayesian joint estimation of multiple graphical models.
\newblock In {\em Advances in Neural Information Processing Systems\/} (2019),
  vol.~32, Curran Associates, Inc.

\bibitem{Giraud14}
{\sc Giraud, C.}
\newblock {\em Introduction to High-Dimensional Statistics}.
\newblock Chapman \& Hall/CRC Monographs on Statistics \& Applied Probability.
  Taylor \& Francis, 2014.

\bibitem{HastieEtAl15}
{\sc Hastie, T., Tibshirani, R., and Wainwright, M.}
\newblock {\em Statistical Learning with Sparsity: The Lasso and
  Generalizations}.
\newblock Chapman \& Hall/CRC Monographs on Statistics and Applied Probability.
  CRC Press, 2015.

\bibitem{LiEtAl15}
{\sc Li, Y., Nan, B., and Zhu, J.}
\newblock Multivariate sparse group lasso for the multivariate multiple linear
  regression with an arbitrary group structure.
\newblock {\em Biometrics. 71\/} (2015), 354--363.

\bibitem{LiEtAl19}
{\sc Li, Z., Mccormick, T., and Clark, S.}
\newblock Bayesian joint spike-and-slab graphical {L}asso.
\newblock In {\em Proceedings of the 36th International Conference on Machine
  Learning\/} (2019), vol.~97 of {\em Proceedings of Machine Learning
  Research}, PMLR, pp.~3877--3885.

\bibitem{LiquetEtAl16}
{\sc Liquet, B., Bottolo, L., Campanella, G., Richardson, S., and Chadeau-Hyam,
  M.}
\newblock {R2GUESS}: A graphics processing unit-based {R} package for
  {B}ayesian variable selection regression of multivariate responses.
\newblock {\em J. Stat. Softw. 69}, 2 (2016), 1--32.

\bibitem{LiquetEtAl17}
{\sc Liquet, B., Mengersen, K., Pettitt, A.~N., and Sutton, M.}
\newblock Bayesian variable selection regression of multivariate responses for
  group data.
\newblock {\em Bayesian Anal. 12}, 4 (2017), 1039--1067.

\bibitem{MaathuisEtAl18}
{\sc Maathuis, M., Drton, M., Lauritzen, S.~L., and Wainwright, M.}
\newblock {\em Handbook of Graphical Models}.
\newblock Chapman \& Hall/CRC Handbooks of Modern Statistical Methods. CRC
  Press, 2018.

\bibitem{MassamWesolowski06}
{\sc Massam, H., and Weso{\l}owski, J.}
\newblock The {M}atsumoto-{Y}or property and the structure of the {W}ishart
  distribution.
\newblock {\em J. Multivariate. Anal. 97\/} (2006), 103--123.

\bibitem{MeinshausenBuhlmann06}
{\sc Meinshausen, N., and B\"uhlmann, P.}
\newblock High-dimensional graphs and variable selection with the {L}asso.
\newblock {\em Ann. Stat. 34}, 3 (2006), 1436--1462.

\bibitem{OkomeEtAl21}
{\sc Okome~Obiang, E., J\'ez\'equel, P., and Pro\"ia, F.}
\newblock A partial graphical model with a structural prior on the direct links
  between predictors and responses.
\newblock {\em ESAIM Probab. Stat. 25\/} (2021), 298--324.

\bibitem{ParkCasella08}
{\sc Park, T., and Casella, G.}
\newblock The {B}ayesian {L}asso.
\newblock {\em J. Am. Stat. Assoc. 103}, 482 (2008), 681--686.

\bibitem{PetrettoEtAl10}
{\sc Petretto, E., Bottolo, L., Langley, S.~R., Heinig, M., McDermott-Roe, C.,
  Sarwar, R., Pravenec, M., H\"{u}bner, N., Aitman, T.~J., Cook, S.~A., and
  Richardson, S.}
\newblock New insights into the genetic control of gene expression using a
  bayesian multi-tissue approach.
\newblock {\em PLOS Comput. Biol. 6}, 4 (2010), 1--13.

\bibitem{RavikumarEtAl11}
{\sc Ravikumar, P., Wainwright, M., Raskutti, G., and Yu, B.}
\newblock High-dimensional covariance estimation by minimizing
  $\ell_1$-penalized log-determinant divergence.
\newblock {\em Electron. J. Stat. 5\/} (2011), 935--980.

\bibitem{RenEtAl15}
{\sc Ren, Z., Sun, T., Zhang, C.~H., and Zhou, H.~H.}
\newblock Asymptotic normality and optimalities in estimation of large
  {G}aussian graphical models.
\newblock {\em Ann. Stat. 43}, 3 (2015), 991--1026.

\bibitem{RothmanEtAl08}
{\sc Rothman, A.~J., Bickel, P.~J., Levina, E., and Zhu, J.}
\newblock Sparse permutation invariant covariance estimation.
\newblock {\em Electron. J. Stat. 2\/} (2008), 494--515.

\bibitem{SohnKim12}
{\sc Sohn, K.~A., and Kim, S.}
\newblock Joint estimation of structured sparsity and output structure in
  multiple-output regression via inverse-covariance regularization.
\newblock In {\em Proceedings of the Fifteenth International Conference on
  Artificial Intelligence and Statistics.\/} (2012), vol.~22 of {\em
  Proceedings of Machine Learning Research}, PMLR, pp.~1081--1089.

\bibitem{WeiEtAl20}
{\sc Wei, R., Reich, B.~J., Hoppin, J.~A., and Ghosal, S.}
\newblock Sparse {B}ayesian additive nonparametric regression with application
  to health effects of pesticides mixtures.
\newblock {\em Statist. Sinica 30\/} (2020), 55--79.

\bibitem{XuGosh15}
{\sc Xu, X., and Ghosh, M.}
\newblock Bayesian variable selection and estimation for {G}roup {L}asso.
\newblock {\em Bayesian Anal. 10}, 4 (2015), 909--936.

\bibitem{XuEtAl16}
{\sc Xu, Z., Schmidt, D.~F., Makalic, E., Qian, G., and Hopper, J.~L.}
\newblock Bayesian grouped horseshoe regression with application to additive
  models.
\newblock In {\em AI 2016: Advances in Artificial Intelligence\/} (2016),
  Springer International Publishing, pp.~229--240.

\bibitem{YangNarisetty20}
{\sc Yang, X., and Narisetty, N.}
\newblock Consistent group selection with bayesian high dimensional modeling.
\newblock {\em Bayesian Anal. 15}, 3 (2020), 909--935.

\bibitem{YuanLin07}
{\sc Yuan, M., and Lin, Y.}
\newblock Model selection and estimation in the {G}aussian graphical model.
\newblock {\em Biometrika. 94}, 1 (2007), 19--35.

\bibitem{YuanZhang14}
{\sc Yuan, X.~T., and Zhang, T.}
\newblock Partial {G}aussian graphical model estimation.
\newblock {\em IEEE. T. Inform. Theory. 60}, 3 (2014), 1673--1687.

\end{thebibliography}

\end{document}